\documentclass[lettersize,journal]{IEEEtran}
\usepackage{tikz}
\usepackage{amsmath, amssymb, amsthm, amsfonts}
\usepackage{filecontents}
\usepackage{booktabs}
\usepackage{tabularx}
\usepackage{multirow}
\usepackage{multicol}
\usepackage{makecell}

\usepackage[colorlinks,
            linkcolor=red,
            urlcolor=black,
            anchorcolor=green,            
            citecolor=cyan]{hyperref}

\ifCLASSOPTIONcompsoc
\usepackage[caption=false, font=normalsize, labelfont=sf, textfont=sf]{subfig}
\else
\usepackage[caption=false, font=footnotesize]{subfig}
\fi

\newcolumntype{P}[1]{>{\centering\hspace{0pt}}p{#1}}
\newcolumntype{Z}{>{\centering\let\newline\\\arraybackslash\hspace{0pt}}X}
\newcolumntype{M}[1]{>{\centering\arraybackslash}m{#1}}

\usepackage{algorithmic}
\usepackage{algorithm}

\newtheorem{theorem}{Theorem}

\newtheorem{definition}{Definition}

\begin{document}

\newcommand{\X}{\mathcal{X}}                
\newcommand{\Xn}{\mathcal{X}^n}             
\newcommand{\Xstar}{\mathcal{X}^*}          

\newcommand{\x}{\mathbf{x}}                 
\newcommand{\y}{\mathbf{y}}
\newcommand{\xt}{x_t}                       
\newcommand{\xprev}{x_1, \dots, x_{t-1}}    
\newcommand{\dist}{d}                       
\newcommand{\Ce}{\mathcal{E}_e}             
\newcommand{\e}{e}                          
\newcommand{\Enc}{\mathsf{Enc}}             
\newcommand{\Dec}{\mathsf{Dec}}             
\newcommand{\msg}{\mathcal{M}}              

\newcommand{\G}{\mathcal{G_\theta}}         
\newcommand{\Ga}{\mathcal{G}_{\mathrm{a}}}  
\newcommand{\B}{\mathcal{C}}                
\newcommand{\Gen}{\mathsf{Gen}}             

\newcommand{\ppi}{\mathcal{P}_{\Pi}}        
\newcommand{\ptheta}{\mathcal{P}_\theta}    
\newcommand{\pmsg}{\mathcal{P}_{\msg}}      
\newcommand{\pxi}{\mathcal{P}_{\Xi}}        

\newcommand{\Dkl}{D_{\mathrm{KL}}}          
\newcommand{\Adv}{\mathsf{Adv}}             
\newcommand{\negl}{\mathrm{negl}}           
\newcommand{\A}{\mathcal{A}}                

\newcommand{\Ha}{\mathrm{H}}                  
\newcommand{\I}{\mathrm{I}}                  
\newcommand{\Hb}{\mathrm{H}_\mathrm{b}}      
\newcommand{\rvM}{M}                         
\newcommand{\rvxi}{\Xi}                     
\newcommand{\rvX}{X}                         



\title{Alkaid: Resilience to Edit Errors in Provably Secure Steganography via Distance-Constrained Encoding}

\author{
Zhihan Cao$^1$,
Gaolei Li$^1$,
Jun Wu$^1$,
Jianhua Li$^1$,
Hang Zhang$^2$, and Mingzhe Chen$^3$
\\
\textsuperscript{1}\emph{Shanghai Jiao Tong University, Shanghai, China}\\
\textsuperscript{2}\emph{Cornell University, Ithaca, NY, USA}\\
\textsuperscript{3}\emph{University of Miami, Coral Gables, FL, USA}
}

\maketitle

\begin{abstract}

While provably secure steganography provides strong concealment by ensuring stego carriers are indistinguishable from natural samples, such systems remain vulnerable to real-world edit errors (e.g., insertions, deletions, substitutions) because their decoding depends on perfect synchronization and lacks error-correcting capability. To bridge this gap, we propose Alkaid, a provably secure steganographic scheme resilient to edit errors via distance-constrained encoding. The key innovation integrates the minimum distance decoding principle directly into the encoding process by enforcing a strict lower bound on the edit distance between codewords of different messages. Specifically, if two candidate codewords violate this bound, they are merged to represent the same message, thereby guaranteeing reliable recovery. While maintaining provable security, we theoretically prove that Alkaid offers deterministic robustness against bounded errors. To implement this scheme efficiently, we adopt block-wise and batch processing. Extensive experiments demonstrate that Alkaid achieves decoding success rates of 99\% to 100\% across diverse error channels, delivers a payload of 0.2 bits per token for high embedding capacity, and maintains an encoding speed of 6.72 bits per second, significantly surpassing state-of-the-art (SOTA) methods in robustness, capacity, and efficiency.

\end{abstract}

\section{Introduction}
Steganography is a fundamental primitive for information hiding, whose objective is to embed confidential messages into apparently benign carriers, such as images, audio signals, or text, so as to enable covert communication~\cite{simmons1984prisoners}. Classical techniques, such as least significant bit (LSB) embedding~\cite{aura1996practical}, realize this objective by modifying the least significant bits of pixel values. More recent advances based on deep neural networks (DNNs), including the synthesis of carriers via generative adversarial networks (GANs)~\cite{li2022gan,ramandi2024vidagan,malik2025hybrid} and adaptive embedding through encoder–decoder frameworks~\cite{bui2023rosteals,ge2025cross,li2023steganography,li2024cover}, have further enhanced both imperceptibility and embedding capacity. Nevertheless, despite sustained methodological progress, existing techniques continue to face two intrinsic limitations:

\textbf{\emph{(i)}} \textbf{Lack of provable security guarantees.} In the presence of increasingly sophisticated steganalysis adversaries, most existing methods rely on heuristic notions of imperceptibility and therefore \emph{fail to provide formal, complexity-theoretic assurances of security,} leaving the statistical indistinguishability of stego carriers from natural samples vulnerable under polynomial-time attacks~\cite{boroumand2018deep}.
\textbf{\emph{(ii)}} \textbf{Insufficient robustness to edit errors.} In real-world deployment, stego objects are frequently exposed to \emph{insertion, deletion, or substitution operations} \textbf{(edit errors)} arising from cropping, reformatting, compression, or transmission artifacts~\cite{editerrorscheng}. Even minor structural perturbations, such as the removal of a single image row or the deletion of a short textual segment, can disrupt the alignment of the embedded message, thereby rendering reliable decoding infeasible and \emph{severely degrading system reliability}.

To enhance security, the paradigm of \textbf{Provably Secure Steganography (PSS)} has been proposed~\cite{kaptchuk2021meteor,ding2023discop,bai2025shimmer,wang2025sparsamp,liao2025framework,bai2025provably}. This framework assumes the existence of a publicly known probability distribution that models the statistics of natural carriers, for instance as induced by large language models (LLMs) or generative image models. The communicating parties share a secret key and employ a pseudorandom generator (PRG) that is \emph{computationally indistinguishable} from a truly random source to drive a deterministic sampling procedure over this distribution. Specifically, the sender encodes a message by using the PRG-guided procedure to deterministically sample a carrier from the distribution, while the receiver, equipped with the same key and PRG, reconstructs the sampling process and recovers the message by identifying the unique sample consistent with the received carrier. Since the PRG is computationally indistinguishable from uniform randomness, the resulting stego carriers are themselves \emph{computationally indistinguishable from genuine samples drawn from the underlying distribution,} thereby achieving provable security in the standard computational sense.

However, despite these strong formal guarantees of concealment, such schemes exhibit pronounced fragility in realistic environments, as common transformations including social media compression, text normalization, and resampling \emph{introduce edit errors that frequently result in decoding failure}. This fragility stems from three structural properties: \textbf{\emph{(i)}} \textbf{strong synchronization dependence,} as decoding requires an exact reproduction of the sender’s generation process; \textbf{\emph{(ii)}} \textbf{lack of error awareness,} such that the receiver cannot detect or localize modifications to the carrier; and \textbf{\emph{(iii)}} \textbf{absence of redundancy or error-correcting structure,} since existing schemes allocate no redundancy in order to maximize embedding capacity.

\emph{Thus, how to achieve \textbf{resilience to edit errors} while maintaining provable security in steganography?}

To address this problem, we propose \textbf{\emph{Alkaid}}, a provably secure steganographic scheme based on \textbf{distance-constrained encoding} that is resilient to edit errors. The core of the system lies in systematically integrating the \emph{minimum distance decoding principle}~\cite{hamming1950error} from error-correcting code theory into the steganographic encoding process.  Guided by this principle, we observe that the robustness of decoding fundamentally depends on a measurable distance (e.g., Hamming distance~\cite{hamming1950error} or Levenshtein distance~\cite{lcvenshtcin1966binary}) between the codewords corresponding to different messages. If the codewords for two distinct messages are too close in the sequence space, even minor edits can easily cause decoding ambiguity.

To this end, we proactively impose a \emph{strict sequence distance constraint} during the encoding phase. Specifically, while generating candidate codewords for messages, the algorithm dynamically ensures that the distance between codewords for different messages is not below a predefined threshold. If the codewords for two different messages fall below this threshold, the system regards them as identical and encodes them into a single message. This design guarantees that, as long as the edit errors do not exceed a certain bound, the receiver can uniquely and correctly recover the original message by applying the minimum distance decoding principle, thereby \emph{achieving tolerance to edit errors}.

Furthermore, through formal proof, we demonstrate that the system maintains provable security while establishing a deterministic robustness guarantee against bounded edit errors, achieving a rigorous unification of security and robustness.

The main contributions of this paper is summarized in the following:
\begin{itemize}
    \item \textbf{Theoretical Framework.} We propose a provably secure distance-constrained encoding algorithm to address decoding failures caused by edit errors. By actively introducing edit distance constraints into the encoding process, the algorithm establishes a theoretical link between provable security and robustness against edit errors. We formally prove it simultaneously achieves information-theoretic security and reliable recovery from bounded edit errors.
    \item \textbf{Efficient Implementation.} To operationalize the theoretical framework, we present Alkaid, a provably secure steganographic scheme that efficiently implements distance-constrained encoding through innovative block-wise processing and parallelization, overcoming the computational bottleneck in encoding and decoding under strict edit distance constraints.
    \item \textbf{Comprehensive Evaluation.} Extensive experimental results demonstrate that Alkaid outperforms existing methods by achieving decoding success rates of 99\% to 100\% across various error channels, delivers a payload of 0.2 bits per token for high embedding capacity, and maintaining an efficient encoding speed of 6.72 bits per second, thereby surpassing state-of-the-art (SOTA) methods in robustness, capacity, and efficiency.
\end{itemize}

\section{Background and Related Works}
\subsection{Preliminaries}

\begin{definition}[Generative Model]
\label{def:generate-model}
Let $\mathcal{\X}$ be a finite discrete sample space, and let $\Xn$ denote the space of sequences of length $n$.
A \textbf{generative model} $\mathcal{\G}$ is a probability distribution $\ptheta$, parameterized by $\theta \in \Theta$, defined over the sequence space $\Xn$.
Such a model can be realized via a deterministic function $\mathrm{Gen}_\theta : \mathcal{Z} \to \Xn$, where $\mathcal{Z}$ is the support of a random variable $\zeta$, satisfying that when $\zeta \sim \mathcal{P}_{\mathcal{Z}}$, the output distribution of $\mathrm{Gen}_\theta(\zeta)$ is exactly $\ptheta$.
\end{definition}

An \emph{autoregressive generative model} is a class of generative models that generate sequences in a stepwise manner. 
Models in this class define a joint distribution over sequences that factorizes as
\begin{equation}
\ptheta(x_1, \dots, x_n) = \prod_{t=1}^n \ptheta(x_t \mid x_1, \dots, x_{t-1}).
\end{equation}
When generating a sequence $\mathbf{x} \in \mathcal{X}^n$, the model successively samples each token according to this factorization:
\begin{equation}
x_t \sim \ptheta(x_t \mid x_1, \dots, x_{t-1}), \quad t = 1, \dots, n.
\end{equation}

Most contemporary large language models (LLMs)~\cite{brown2020language, TheC3,deepseek_2025,geminiteam2025} adopt this generative paradigm, where the subsequent content can typically be controlled by specifying a prefix (or history).

\begin{definition}[Edit Error Channel]
\label{def:edit-channel}
An \textbf{edit error channel} $\Ce$ randomly maps an input sequence $\x \in \Xn$ to an output sequence $\tilde{\x} \in \X^\ast$.
For any input sequence $\x$, the channel independently performs a basic edit operation (insertion, deletion, or substitution) at each position with probability $e$, and satisfies the following expected edit distance constraint:
\begin{equation}
    \mathbb{E}[\dist(\x, \tilde{\x})] \leq e \cdot n,
\end{equation}
where the \textbf{edit distance} (Levenshtein distance) $\dist(\mathbf{a}, \mathbf{b})$ denotes the minimum number of basic edit operations required to transform sequence $\mathbf{a}$ into sequence $\mathbf{b}$, and $e \in [0,1]$ is referred to as the error rate of the channel.
\end{definition}

\begin{definition}[Pseudorandom Generator]
Let $\lambda$ be the security parameter, and let $G: \{0,1\}^\lambda \to \{0,1\}^{\ell(\lambda)}$ be a deterministic polynomial-time algorithm, where $\ell(\cdot)$ is a polynomial and $\ell(\lambda) > \lambda$.
For any probabilistic polynomial-time (PPT) distinguisher $\A$, its advantage is defined as
\begin{equation}
\Adv_{\A,G}(\lambda) = 
\left|
\begin{aligned}
\Pr_{\mathsf{sk} \gets \{0,1\}^\lambda}\big[\A(1^\lambda, G(\mathsf{sk})) = 1\big] \\
{} - \Pr_{\mathcal{S} \gets \{0,1\}^{\ell(\lambda)}}\big[\A(1^\lambda, \mathcal{S}) = 1\big]
\end{aligned}
\right|,
\end{equation}
where the probabilities are taken over the uniform random choice of $\mathsf{sk}$ and $\mathcal{S}$, as well as the internal randomness of $\A$.
Then, $G$ is called a \textbf{pseudorandom generator (PRG)} if for every PPT distinguisher $\A$, there exists a negligible function $\negl(\cdot)$ such that
$
\Adv_{\A,G}(\lambda) \leq \negl(\lambda).
$
\end{definition}

\begin{definition}[Generative Steganographic Scheme]
\label{def:stegosystem}
A \textbf{generative steganographic scheme} $\Pi = (\G, \msg, \Xi, \Enc, \Dec)$ consists of the following components:
\begin{itemize}
    \item a \textbf{generative model} $\G$, as defined in Definition~\ref{def:generate-model}, whose output distribution is denoted by $\ptheta$;
    \item a \textbf{message space} $\msg$, which is a finite set of messages;
    \item an \textbf{encoding parameter space} $\Xi$, which is a collection of control parameters shared by the encoder and the decoder;
    \item an \textbf{encoder} $\Enc: \msg \times \Xi \to \Xn$ such that, for any $m \in \msg$ and $\xi \in \Xi$,
    $
    \Enc(m,\xi) = \Gen_\theta\big(f(m,\xi)\big),
    $
    where $f: \msg \times \Xi \to \mathcal{Z}$ is a deterministic function mapping the message and encoding parameters to the input space of the generative model, and $\Gen_\theta$ is the sampling function of the generative model;
    \item a \textbf{decoder} $\Dec: \Xn \times \Xi \to \msg$ satisfying that, for any $m \in \msg$ and $\xi \in \Xi$,
    $
    \Dec(\Enc(m,\xi), \xi) = m.
    $
\end{itemize}
\end{definition}

\begin{definition}[Information-Theoretic Security]
\label{def:information-theoretic-security}
Let the generative steganographic scheme $\Pi$ be as defined in Definition~\ref{def:stegosystem}, where the message $m$ is distributed according to $\pmsg$ and the encoding parameter $\xi$ is distributed according to $\pxi$.
Let $\ptheta$ denote the probability distribution induced by the generative model $\G$.
The output distribution of the steganographic scheme, denoted by $\ppi$, is defined as
\begin{equation}
\ppi(\x) = \Pr_{\substack{ m \sim \pmsg\\ \xi \sim \pxi }}[\Enc(m,\xi) = \x], \quad \forall \x \in \Xn,
\end{equation}
where $\Dkl(\ppi \parallel \ptheta)$ denotes the Kullback--Leibler (KL) divergence from $\ppi$ to $\ptheta$, defined as
\begin{equation}
\Dkl(\ppi \parallel \ptheta) = \sum_{\x \in \Xn} \ppi(\x) \log \frac{\ppi(\x)}{\ptheta(\x)},    
\end{equation}
which quantifies the  distinguishability between the stego output distribution $\ppi$ and the cover distribution $\ptheta$.
If $\Dkl(\ppi \parallel \ptheta) \leq \epsilon$, then $\Pi$ is said to be \textbf{$\epsilon$-secure}.
If $\Dkl(\ppi \parallel \ptheta) = 0$, then $\Pi$ is said to be \textbf{perfectly secure}.
\end{definition}

\begin{definition}[Computational Security]
\label{def:computational-indistinguishability-security}
Let $\lambda$ be the security parameter, and let the steganographic scheme $\Pi$ be as defined in Definition~\ref{def:stegosystem}, with all its parameters being functions of $\lambda$.
For any probabilistic polynomial-time (PPT) distinguisher $\A$, its advantage is defined as

\begin{equation}
\Adv_{\A,\Pi}(\lambda) =
\left|
\begin{aligned}
\Pr_{\substack{m \sim \pmsg \\ \xi \sim \pxi}}
\big[\A(1^\lambda, \Enc(m,\xi)) = 1\big] \\
{} - \Pr_{\x \sim \ptheta}
\big[\A(1^\lambda, \x) = 1\big]
\end{aligned}
\right|,
\end{equation}
and for all PPT distinguishers $\A$, if there exists a negligible function $\negl(\cdot)$ such that
$
\Adv_{\A,\Pi}(\lambda) \leq \negl(\lambda),
$
then $\Pi$ is said to be \textbf{computationally secure}.
\end{definition}

\subsection{Related Works}

\subsubsection{Theoretical Secure Steganography}
The formal security definition of steganography serves as the theoretical cornerstone for its development. Traditional digital steganographic methods primarily pursue high embedding capacity and empirical imperceptibility but lack rigorous security proofs~\cite{wangBreakGenerativeSteganography2026,zhang-etal-2021-provably}. Cachin~\cite{cachin1998information} defined the security of steganography from an information-theoretic perspective, using the Kullback-Leibler divergence~\cite{kullback1951information} between the cover distribution and the stego distribution as a measure (as formalized in Definition~\ref{def:information-theoretic-security}), thereby laying the theoretical foundation for steganalysis. Hopper et al.~\cite{hopper2002provably} further proposed a provably secure framework based on computational complexity theory (as formalized in Definition~\ref{def:computational-indistinguishability-security}), where security is defined through a probabilistic polynomial-time distinguishing game: a system is considered secure if an adversary cannot distinguish with significant advantage between outputs from the natural channel distribution \(O_D\) and those from the steganographic encoder \(\text{Encode}_D\). These works established strict security objectives and evaluation paradigms for all subsequent provably secure steganographic schemes from both information-theoretic and computational complexity perspectives.

\subsubsection{Practical and Secure Steganography}
Within the above theoretical framework, researchers have endeavored to construct efficient and secure practical systems. Early rejection sampling and arithmetic coding schemes, while satisfying security definitions, were inefficient and difficult to apply in practice~\cite{hopper2002provably,von2004public,backes2005public}. Meteor~\cite{kaptchuk2021meteor} realized a provably secure practical system on language generation models by pseudorandomizing secret information to avoid the randomness reuse problem, though its embedding capacity was limited by per-step entropy wastage. Discop~\cite{ding2023discop} significantly increased capacity by constructing multiple "distribution copies" and sampling using a balanced binary tree, but introduced additional computational overhead. Shimmer~\cite{bai2025shimmer} introduced an entropy collection mechanism that recycles residual entropy from the previous step into subsequent encoding, further optimizing the balance between capacity and efficiency. Sparsamp~\cite{wang2025sparsamp} employed a sparse sampling strategy to pursue higher time efficiency while maintaining security. Despite continuous progress in security and capacity, these schemes share a fundamental flaw: their encoding and decoding processes must remain perfectly synchronized. Any editing operation on the carrier sequences disrupts this synchronization context, preventing the decoder from reconstructing the encoding path and leading to decoding failure. This acute vulnerability to edit errors severely restricts the practical application of the aforementioned schemes in real, unreliable channels.

\begin{figure*}[!ht]%
\centering
\includegraphics[width=0.95\linewidth]{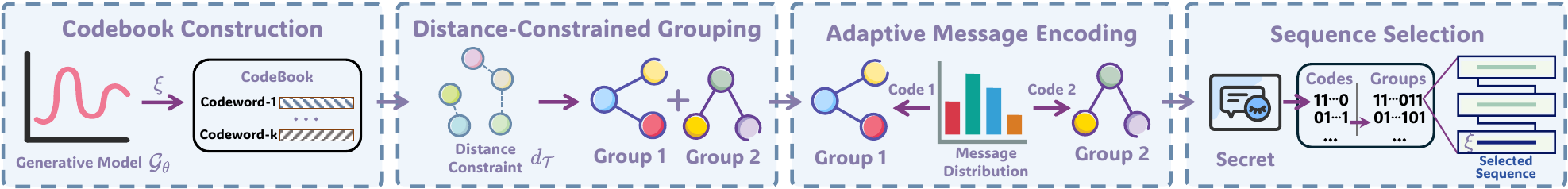}
    \caption{The process of distance-constrained encoding involves four mainsteps: \textbf{\emph{(i)}} \textbf{Codebook Construction,} where a generative model $\G$ utilizes encoding parameters $\xi$ to produce $k$ sequences as codewords; \textbf{\emph{(ii)}} \textbf{Distance-Constrained Grouping,} in which codewords with edit distances below a given threshold $d_{\mathcal{T}}$ are grouped together; \textbf{\emph{(iii)}} \textbf{Adaptive Message Encoding,} which dynamically allocates codes to represents specific message within each group; \textbf{\emph{(iv)}} \textbf{Sequence Selection,} where the specific secret message and encoding parameters $\xi$ together determine the unique sequence to serve as the stego carrier.}
    \label{fig.dce}
\end{figure*}

\subsubsection{Robust and Secure Steganography}
Robustness is a key attribute for ensuring the practical reliability of steganographic systems. Traditional robust steganography research primarily focuses on resisting conventional signal processing operations like JPEG compression, noise addition, and resampling~\cite{zeng2022improving,qiao2021robust,yu2023cross,yang2025novel}. Its design is usually based on heuristic methods or signal processing techniques and is not situated within a provably secure formal framework. Consequently, its security is difficult to guarantee. In recent years, some studies have begun to explore robustness issues while aiming for provable security. For example, Zamir~\cite{zamir2024undetectable} explored combining pseudorandom error-correcting codes~\cite{christ2024pseudorandom} with steganography to enhance undetectability but did not conduct an in-depth analysis of its tolerance to edit errors. The scheme proposed by Bai et al.~\cite{bai2025provably} for asymmetric resource scenarios mainly targets robustness against substitution errors. STEAD~\cite{qistead} achieved robustness against edit errors on diffusion models, but its method heavily relies on specific model structures, has limited generality, and suffers from low embedding capacity. Thus, within the paradigm of provably secure steganography, how to systematically introduce robustness against edit errors remains a critical yet insufficiently addressed problem.

\section{Distance-Constrained Encoding}

How can we construct a steganographic scheme that is robust against edit errors? A feasible approach is based on the \emph{minimum distance decoding} principle: given a received sequence $\tilde{\x}$ and the known encoding parameter $\xi$, the decoder selects the message corresponding to the valid codeword that is closest to $\tilde{\x}$ in terms of edit distance, as defined by
\begin{equation}
\Dec(\tilde{\x}, \xi) = \arg\min_{m \in \msg} \dist\bigl(\tilde{\x}, \Enc(m, \xi)\bigr).
\end{equation}

However, the effectiveness of minimum distance decoding relies on an important assumption: \emph{codewords corresponding to different messages must be sufficiently separated}. If two codewords are too close, channel errors can easily make the received sequence closer to another codeword, resulting in a decoding error. Therefore, to ensure robustness, we must impose an explicit distance constraint during encoding, namely, the edit distance between any two codewords corresponding to different messages must be no smaller than a threshold $d_\mathcal{T}$:
\begin{equation}
 \dist\bigl(\Enc(m_i, \xi), \Enc(m_j, \xi)\bigr) > d_\mathcal{T}, \quad \forall m_i \neq m_j \in \msg, \ \xi \in \Xi.   
\end{equation}

To achieve robustness while maintaining information-theoretic security, we propose \textbf{distance-constrained encoding}. The key idea is to assign codewords that are sufficiently far apart for different messages by imposing an edit-distance constraint during encoding, so that the decoder can reliably correct errors introduced by the channel. The procedure can be summarized in the following four steps:
\subsection{Main Steps}

\noindent \textbf{Codebook Construction.}
Given the encoding parameter $\xi = (\zeta_1, \dots, \zeta_k, \eta) \in \Xi$, where $\zeta_1, \dots, \zeta_k$ are used to generate codewords and $\eta$ controls intra-group selection, we use the generative model $\G$ to produce the initial codebook $\B^\xi = \{\x^{(1)}, \dots, \x^{(k)}\}$, where $\x^{(i)} = \Gen_\theta(\zeta_i)$.

\noindent \textbf{Distance-Constrained Grouping.}
Set a minimum distance threshold $d_\mathcal{T} > 0$. Partition $\B^\xi$ into several groups so that any two sequences from different groups satisfy $\dist(\x, \x') > d_\mathcal{T}$. The resulting $K$ mutually disjoint groups are $\mathcal{O}^\xi_1, \dots, \mathcal{O}^\xi_K$, with the property that for any $i \neq j$, $\dist(\x, \x') > d_\mathcal{T}$ for all $\x \in \mathcal{O}^\xi_i, \x' \in \mathcal{O}^\xi_j$.

\noindent \textbf{Adaptive Message Encoding.} Assign a unique message $m$ to each group $\mathcal{O}^\xi_m$, giving $|\msg^\xi| = K$. To maintain equal probability usage of codewords, the occurrence probability of $m$ should be proportional to the size of its group, which can be written as $\pmsg^\xi(m) = \frac{|\mathcal{O}^\xi_m|}{k}$. This requires adaptive message encoding depending on the grouping situation.

\noindent \textbf{Sequence Selection}. Using $\eta$ in $\xi$, a deterministic function $h: \msg \times \mathcal{H} \to \Xn$ selects the codeword sequence, where $\mathcal{H}$ is the value space of $\eta$. The function satisfies: for any message $m$ and encoding parameter $\xi$, when $\eta \sim \mathcal{P}_{\mathcal{H}}$, $h(m,\eta)$ preserves the original probability distribution over the group $\mathcal{O}^\xi_m$. That is, for any $\x \in \mathcal{O}^\xi_m$, $\Pr_{\eta \sim \mathcal{P}_{\mathcal{H}}}[h(m,\eta) = \x] = \frac{t_\x(\xi)}{|\mathcal{O}^\xi_m|}$, where $t_\x(\xi)$ denotes the number of occurrences of $\x$.

As depicted in Fig.~\ref{fig.dce}, based on the above steps, the encoder $\Enc$ determines and transmits the codeword from the corresponding group according to the message \( m \) and the parameter \( \eta \). During decoding, the  edit distance threshold \( d_{\mathcal{T}} \) enforced between groups provides a clear discriminative boundary for minimum distance decoding, thereby theoretically ensuring robustness against a certain degree of channel interference. The entire scheme transforms the structured codebook design into a foundation for reliable transmission.

\subsection{Security}

\begin{theorem}
\label{thm:distance-constraint-security}
Let $\Pi$ be a steganographic scheme in which the encoder $\Enc$ uses distance-constrained encoding, and its encoding parameter $\xi = (\zeta_1, \dots, \zeta_k, \eta)$ consists of mutually independent elements sampled randomly from their respective spaces, and all message spaces $\msg^\xi$ follow the distribution $\pmsg^\xi(m) = \frac{|\mathcal{O}^\xi_m|}{k}$, then the steganographic scheme is \textbf{perfectly secure}. That is, its output distribution satisfies $\ppi = \ptheta$.
\end{theorem}

The theorem above demonstrates that steganographic text (stegotext) generated using the distance-constrained encoding algorithm is \emph{indistinguishable} from the model's original random sampling output. The complete proof is provided in Appendix~\ref{appendix:proof_theoreticsecurity}.

\subsection{Robustness}

\begin{theorem} 
\label{thm:distance-constraint-robustness} 
Let $\Pi$ be a steganographic scheme in which the encoder $\Enc$ uses distance-constrained encoding and the decoder $\Dec$ employs minimum distance decoding. Suppose the edit-error channel $\Ce$ satisfies Definition~\ref{def:edit-channel} with error probability $e$. If the distance constraint satisfies $n >d_\mathcal{T} > 2 e n$, then the system is  \textbf{$(e,\delta)$-robust} with
\begin{equation}
\delta  \triangleq \exp\!\left(- \frac{(d_\mathcal{T} - 2 e n)^2}{2 n} \right), 
\end{equation}
and for any encoding parameter $\xi$ and any message $m \in \msg^\xi$, 
\begin{equation}
\Pr[\Dec(\Ce(\Enc(m,\xi)), \xi) \neq m] \le \delta.   
\end{equation}

\end{theorem}

This theorem guarantees a low upper bound for the decoding failure rate when the distance constraint \( d_\mathcal{T} > 2 e n \) . The complete proof is provided in Appendix~\ref{appendix:dce-roubust}. Moreover, in practical deployment, parameters such as the sample size $k$ and the specific algorithm implementation can further reduce the actual failure rate \emph{below this theoretical threshold}. We will analyze the resulting practical robustness in Appendix~\ref{appendix:analysisofrobust}.

\section{Main Design}

From Theorems \ref{thm:distance-constraint-security} and \ref{thm:distance-constraint-robustness}, it follows that a steganographic scheme adopting distance-constrained encoding combined with minimum distance decoding can theoretically guarantee \emph{perfect security and robustness} simultaneously. In practice, however, this method encounters three major obstacles:

\textbf{\emph{(i)}} \textbf{Inefficient generation.} The codebook must be constructed by independently sampling $k$ complete sequences, which requires multiple invocations of the generative model and incurs significant computational overhead.  
    
\textbf{\emph{(ii)}} \textbf{Challenges in parameter synchronization.} The encoding parameter $\xi = (\zeta_1, \dots, \zeta_k, \eta)$ must consist of strictly independent elements following their respective distributions. Achieving such large-scale parameter sharing and synchronization is difficult in distributed environments. 
    
\textbf{\emph{(iii)}} \textbf{Message space dependency.} The message probability $\pmsg^\xi(m) = \frac{|\mathcal{O}^\xi_m|}{k}$ depends on both the generative model $\G$ and the encoding parameter $\xi$. Any change in $\xi$ requires re-encoding the messages to maintain consistency with the corresponding message distribution.

To address the aforementioned issues, we propose a series of methods for efficiently implementing distance-constrained encoding, which significantly improves its practical applicability. The core approaches are elaborated in detail as follows.

\subsection{Core Approaches}

\textbf{Codebook Construction.} To avoid explicit transmission and synchronization of the encoding parameter $\xi$, the encoder and decoder share a secret key $\mathsf{sk}$ and derive the required randomness on the fly using a cryptographically secure pseudorandom generator (PRG). In practical implementation, we can segment the PRG randomly generated sequence into blocks and process them into encoding parameters. Moreover, since most state-of-the-art large language models are autoregressive and generate tokens sequentially, the steganographic scheme must also construct the codebook in a token-by-token manner. Accordingly, the proposed scheme builds the codebook via sequential generation (Algorithm~\ref{alg:cbc}). 

The procedure is as follows. At step $t$, a random value $r_t \in [0,1)$ is first generated using the PRG. The output token is then determined via inverse transform sampling (InvCDF), yielding $x_t = \mathrm{InvCDF}(r_t)$. Repeating this process produces $k$ sequences of length $n_l$, which together form the codebook.

However, even when a cryptographically secure PRG is employed, its outputs are not statistically equivalent to samples drawn independently from the true distribution, but are only computationally indistinguishable from them. Consequently, the system can no longer achieve perfect security in the information-theoretic sense, and the subsequent security analysis is therefore conducted under the notion of computational security.

\begin{algorithm}[!t]
\caption{CodebookConstruction}
\label{alg:cbc}
\renewcommand{\algorithmicrequire}{\textbf{Input:}}
\renewcommand{\algorithmicensure}{\textbf{Output:}}
\begin{algorithmic}[1]
\REQUIRE Model $\G$, history $h$, sample size $k$, block length $n_l$, \textsc{PRG} $\mathcal{R}$.
\ENSURE Codebook $\B = \{b_1,\dots,b_{k}\}$.

\STATE $\B \leftarrow \emptyset$

\FOR{each $i \in [1,k]$}
    \STATE $c_i \leftarrow \emptyset$
    \STATE $\textsf{ctx} \leftarrow h$ 

    \FOR{each step $t \in [1,n_l]$}
        \STATE $p(\cdot)\leftarrow \textsc{NextTokenDist}(\G,\textsf{ctx})$ 
        \STATE $r_{i,t} \leftarrow \mathcal{R}(i,t)$
        \STATE $x_t \leftarrow \textsc{InvCDF}(p,r_{i,t})$ 
        \STATE $c_i \leftarrow c_i \,||\, x_t$; \quad $\textsf{ctx} \leftarrow \textsf{ctx} \,||\, x_t$
    \ENDFOR

    \STATE $\B \leftarrow \B \cup \{c_i\}$
\ENDFOR
\RETURN $\B$
\end{algorithmic}
\end{algorithm}

\noindent \textbf{Distance-Constrained Grouping.} As shown in Algorithm~\ref{alg:dcg}, this grouping procedure assigns all elements whose pairwise edit distance does not exceed the threshold \(d_{\text{min}}\) to the same group. In practice, this can be efficiently implemented using the disjoint set union (DSU). By iteratively evaluating the edit distance between pairs of elements and merging those that satisfy the constraint into the same set, the algorithm ultimately produces a collection of disjoint sets, each corresponding to a valid group under the distance constraint, thereby completing the overall partition efficiently.

\begin{algorithm}[!t]
\caption{DistanceConstraintGrouping}
\label{alg:dcg}
\renewcommand{\algorithmicrequire}{\textbf{Input:}}
\renewcommand{\algorithmicensure}{\textbf{Output:}}
\begin{algorithmic}[1]
\REQUIRE CodeBook $\B=\{c_i\}_{i=1}^{k}$, threshold $d_\mathcal{T}$.
\ENSURE Groups $\mathcal{O}=\{O_1,\dots,O_m\}$.

\STATE  \textsc{DSU.init()} 
\FOR{all $1 \le i < j \le k$}
    \IF{$\dist(b_i,b_j) \le d_\mathcal{T}$}
        \STATE \textsc{Union}$(i,j)$
    \ENDIF
\ENDFOR
\STATE $\mathcal{O} \leftarrow$ \textsc{DSU}.components
\RETURN $\mathcal{O}$
\end{algorithmic}
\end{algorithm}

\noindent \textbf{Adaptive Message Encoding.}
Under fixed encoding parameters, distance-constrained encoding induces a concrete message probability distribution $\pmsg^\xi(m)$. To realize a corresponding encoding scheme, one may design a variable-length code over the message space so that the original distribution is preserved under encoding. Since $\pmsg^\xi(m)=\frac{|\mathcal{O}^\xi_m|}{k}$ and $|\mathcal{O}^\xi_m|$ is an integer, each message probability is an integer multiple of $\frac{1}{k}$. In particular, when $k=q^\ell$, the entire probability space can be partitioned into $q^\ell$ equiprobable units, and each unit can be assigned a unique $\ell$-symbol codeword over an alphabet of size $q$. Each message $m$ corresponds to $|\mathcal{O}^\xi_m|$ such units. And if these units share a common prefix, that prefix can be used as the encoding of message $m$.

For example, let sample size $k=16$ and consider the distribution $[3/4, 1/8, 1/16, 1/16]$, so that the probability space is divided into $16$ units. Under binary encoding, the first message has probability $12/16$ and thus occupies more than half of the units. As a result, no non-empty prefix can be assigned to it, since any such prefix would correspond to a probability strictly smaller than the original message probability; therefore, it must be represented by the empty string $\varnothing$. The second message is then assigned the prefix $110$, corresponding to the next $2$ units. The last two messages each occupy a single unit and are encoded as $1110$ and $1111$, respectively. The resulting encoding is $[\varnothing, 110, 1110, 1111]$. Fig.~\ref{fig.AME} vividly illustrates this example. This construction assigns a common prefix to all units covered by the same message and employs longest-prefix matching at the decoder, thereby ensuring that the induced encoding distribution coincides with the original distribution. For a distribution such as $[3/8, 3/8, 1/4]$, the same allocation procedure yields the encoding $[0,\varnothing,\varnothing]$. In this case, the latter two symbols share an identical encoding and are explicitly merged into a single group, which then participates in subsequent processing as a whole. 

\begin{figure}[!t]%
\centering
\includegraphics[width=0.95\linewidth]{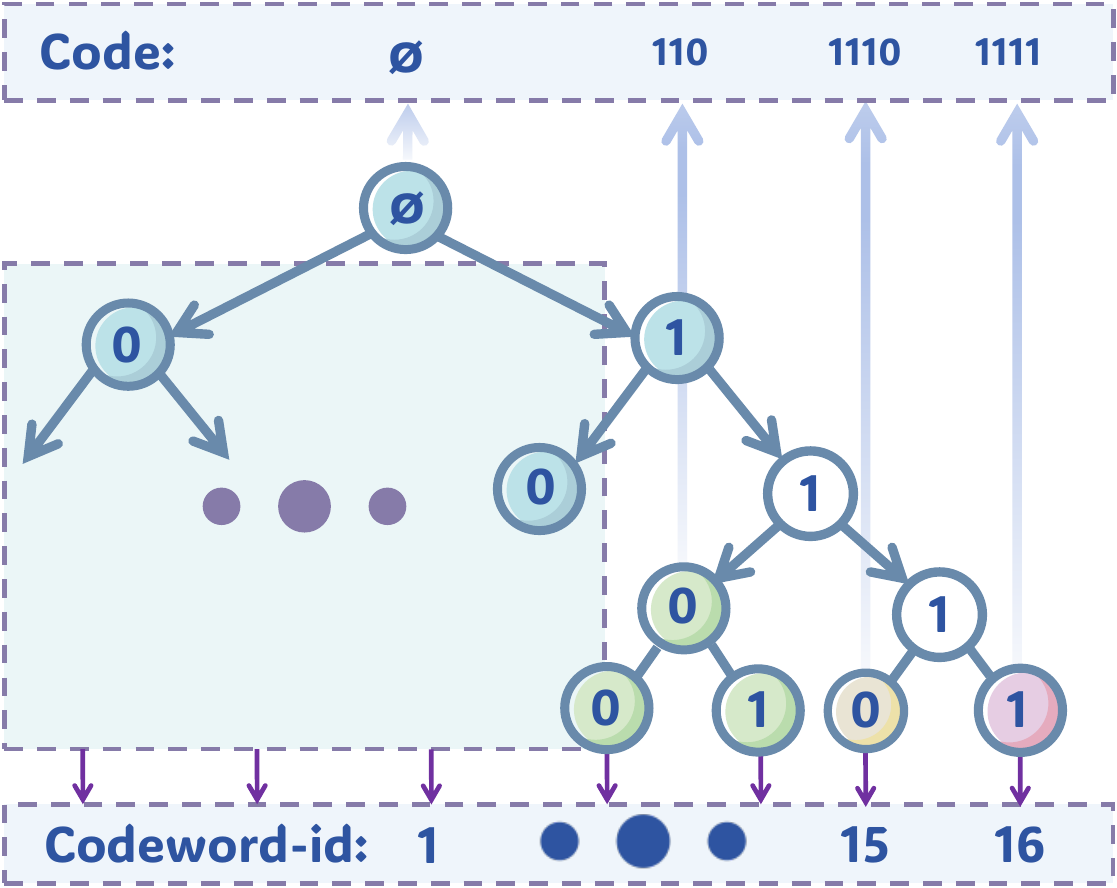}
    \caption{An example of adaptive message encoding: Given a group distribution $[3/4, 1/8, 1/16, 1/16] $, we construct a binary tree of depth 5. The last layer contains 16 nodes, each representing a codeword. The encoding for each group is the common prefix of its codewords, resulting in the final codes: $[\varnothing, 110, 1110, 1111] $.}
    \label{fig.AME}
\end{figure}

The corresponding implementation is given in Algorithm~\ref{alg:ac}. The algorithm achieves adaptive encoding for non-uniform message distributions by dynamically assigning binary prefixes to message groups of different sizes. Its core mechanism maintains a complete binary tree of depth \(\log_2 k\), where each leaf corresponds to a codeword and is initially unoccupied. The size \(|O|\) of a message group \(O\) determines the number of leaves it must occupy. Groups are processed in descending order of size. For the current group, the algorithm searches for a node \(v\) whose subtree contains at least \(|O|\) unoccupied leaves, preferring nodes at greater depth so as to obtain longer prefixes. It then marks \(|O|\) leaves within that subtree as occupied and assigns the path label from the root to \(v\) (a sequence of 0/1 symbols) as the prefix encoding for the group. This process is repeated until all groups have been assigned.

\begin{algorithm}[!t]
\caption{AdaptiveCoding}
\label{alg:ac}
\renewcommand{\algorithmicrequire}{\textbf{Input:}}
\renewcommand{\algorithmicensure}{\textbf{Output:}}
\begin{algorithmic}[1]
\REQUIRE Groups $\mathcal{O}=\{O_1,\dots,O_m\}$, sample size $k$.
\ENSURE Message $\msg$.

\STATE $\mathcal{T} \leftarrow \textsc{BuildTree}(k)$
\STATE $\textsc{$\mathcal{T}$.InitLeaf()}$

\STATE $\mathcal{O}^* \leftarrow \textsc{Sort}(\mathcal{O})$
\STATE $\msg \leftarrow \emptyset$

\FOR{each group $O \in \mathcal{O}^*$}
    \STATE $q \leftarrow |O|$
    \STATE $v \leftarrow \textsc{FindAllocNode}(\mathcal{T}, q)$
    \STATE $p \leftarrow \textsc{PrefixOfNode}(v)$ 
    \STATE $\mathcal{L}_O \leftarrow \textsc{SelectLeaves}(v, q)$ 
    \STATE $\textsc{MarkUsed}(\mathcal{L}_O)$
    \STATE $\msg \leftarrow \msg \cup \{(O \mapsto p)\}$
\ENDFOR

\RETURN $\msg$
\end{algorithmic}
\end{algorithm}

\begin{figure*}[!ht]%
\centering
\includegraphics[width=0.95\linewidth]{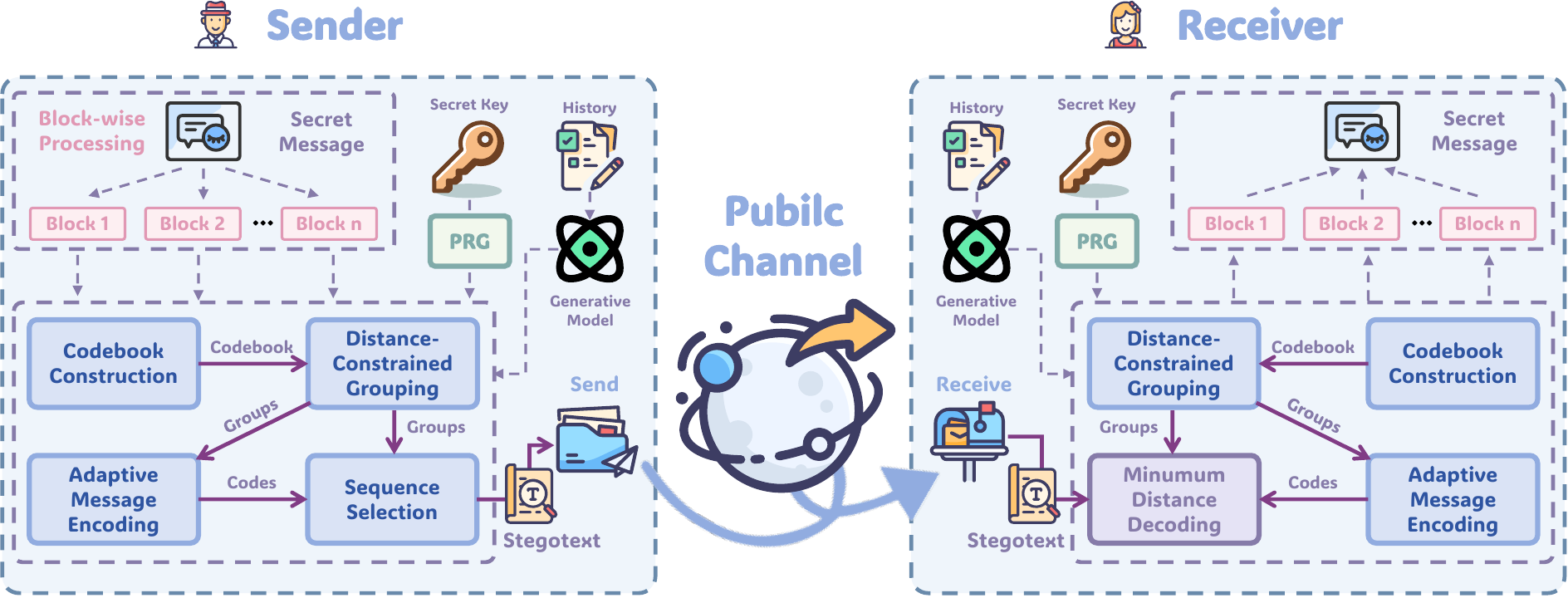}
    \caption{Overview of Alkaid. This scheme requires both the sender and receiver to use the same generative model, historical context, and secret key. On the sender side, the secret message is divided into blocks, each of which undergoes distance-constrained encoding until the entire message is embedded into the stegotext. On the receiver side, the same distance-constrained encoding process is repeated to reconstruct the codebook. Each message block is then recovered using minimum distance decoding, and finally reassembled to obtain the original secret message.}
    \label{fig.overview}
\end{figure*}

\noindent \textbf{Sequence Selection.}
As described in Algorithm~\ref{alg:ss}, the determination of the output sequence proceeds in two stages. 
First, each leaf node represents a probability unit. Given the secret message with the same code length, a unique leaf node is identified, which in turn determines the group to which it belongs. 
Second, the intra-group selection is carried out by generating a random value via a PRG and using it as the input to inverse transform sampling (InvCDF), thereby deterministically producing a concrete sequence. 
By the principle of inverse transform sampling, as long as the PRG output is uniformly random, the final output follows the prescribed distribution over its corresponding group $\mathcal{O}^\xi_m$, which ensures that the resulting procedure constitutes a valid function~$h$.

The introduction of the PRG serves a critical purpose. Sequences within the same group are mutually similar, and without synchronized randomness, the decoder may be unable to uniquely recover the specific sequence chosen by the sender. \emph{By sharing a common seed to synchronize the PRG, the encoder and decoder can perform the same intra-group selection consistently, thereby eliminating ambiguity in decoding.}

\begin{algorithm}[!t]
\caption{SequencesSelection}
\label{alg:ss}
\renewcommand{\algorithmicrequire}{\textbf{Input:}}
\renewcommand{\algorithmicensure}{\textbf{Output:}}
\begin{algorithmic}[1]
\REQUIRE Message $\msg$, secret bits $s$, \textsc{PRG} $\mathcal{R}$.
\ENSURE Selected sequence $t$, consumed prefix $p$.

\STATE $b \leftarrow \textsc{Head}(s)$ 
\STATE $O^\star \leftarrow \textsc{UnitMatch}(\msg,b)$ 
\STATE $p \leftarrow \msg[O^\star]$ 
\STATE $t \leftarrow \textsc{InvCDF}(O^\star,\mathcal{R}(\cdot))$ 
\RETURN $(t,p)$
\end{algorithmic}
\end{algorithm}

\subsection{Overview}

\textbf{Block-wise Processing.}
Constructing the codebook requires independently sampling $k$ complete sequences in advance, which entails repeated invocations of the generative model and incurs substantial computational cost. When the sequence length satisfies $n \gg 0$ and $k = 2^\ell$, the resulting computational complexity is $O(n \cdot 2^\ell)$, which grows exponentially and is often prohibitive in practice. 

To address this issue, a  block-wise strategy can be adopted. Specifically, the $\ell$-bit message is partitioned into multiple smaller blocks, each containing only $b$ bits (for instance, $b = 1$). For each block, only $2^b$ candidate sequences are generated, and the overall output sequence is produced by generating and concatenating these blocks sequentially. If each block generates a subsequence of length $n_l$, the total complexity is reduced to 
$O\!\left( n_l \cdot 2^b \cdot \frac{\ell}{b} \right)$. 
Choosing small values of $n_l$ and $b$ can therefore substantially improve generation efficiency. 

However, it should be noted that under an error-correction design based on minimum distance decoding, reducing the block length also correspondingly weakens the error-correction capability. Therefore, practical deployment requires a careful trade-off between computational cost and robustness, tailored to the specific application scenario.

\noindent \textbf{Encoding Procedure.}
The encoding process follows the distance-constrained encoding paradigm. As illustrated in the Fig.~\ref{fig.overview}, the steganographic workflow operates in a block-wise manner. The detailed procedure is given in Algorithm~\ref{alg:encode}. Here, we assume the secret message is a bit string. 

First, the $\ell$-bit secret message to be transmitted is partitioned into multiple blocks of length $b$ bits. For each block, conditioned on the current historical context, the system generates $k = 2^b$ candidate sequences of length $n_l$ to form a codebook, which is then partitioned into groups according to the edit-distance constraint. 
Next, adaptive message encoding is applied to assign an identifier to each group, and the group corresponding to the current secret message block is selected based on its content. Subsequently, within the chosen group, a specific output sequence is deterministically selected using pseudorandomness. The selected sequence is appended to the end of the historical context and serves as the input for the next generation round. Finally, the system computes the amount of secret information encoded in the current block and removes it from the remaining message to be embedded. Bits that fail to be encoded this time will be carried over to the next processing cycle. This procedure is iterated until all secret information has been embedded into the generated sequence.

\begin{algorithm}[!t]
\caption{Encode}
\label{alg:encode}
\renewcommand{\algorithmicrequire}{\textbf{Input:}}
\renewcommand{\algorithmicensure}{\textbf{Output:}}
\begin{algorithmic}[1]
\REQUIRE Model $\G$, history $h$, secret bits $s$, sample size $k$, distance constraint $d_\mathcal{T}$, block length $n_l$, seed $\mathcal{S}$.
\ENSURE Encoded token sequence $T$.
\STATE $T \leftarrow \emptyset$
\STATE $\mathcal{R} \leftarrow \textsc{PRG.Init($\mathcal{S}$)}$
\WHILE{$len(s) > 0$}
    \STATE $\B \leftarrow \textsc{CodebookConstruction}(\G,h,k,n_l,\mathcal{R})$ 
    \STATE $\mathcal{O} \leftarrow \textsc{DistanceConstraintGrouping}(\B,d_\mathcal{T})$ 
    \STATE $\msg \leftarrow \textsc{AdaptiveCoding}(\mathcal{O})$ 
    \STATE $(t,p) \leftarrow \textsc{SequencesSelection}(\mathcal{O},\msg,s,\mathcal{R})$ 
    \STATE $T \leftarrow T \,||\, t; \quad h \leftarrow h\, || \, t $ 
    \STATE $s \leftarrow \textsc{removePrefix}(s,p)$
\ENDWHILE
\RETURN $T$
\end{algorithmic}
\end{algorithm}

\noindent \textbf{Decoding Procedure.}
The decoding process is likewise executed in a block-wise manner and follows the minimum-distance decoding principle. The detailed procedure is given in Algorithm~\ref{alg:decode}. The system processes each received block sequentially and reproduces the same codebook construction procedure as in the encoding stage, thereby obtaining the set of candidate sequences (the codebook) corresponding to the current block. It then computes the distances between the received sequence and all sequences in the codebook, and selects the group containing the closest sequence as the decoding result. 

Owing to the use of a PRG for synchronizing sequence selection, repeating the same procedure ensures \emph{consistency of the historical context} between the encoder and decoder. To address sequence misalignment caused by insertions and deletions, the decoder introduces a sliding window of size $w$, enlarging the comparison range to mitigate the impact of such shifts. Meanwhile, within each block, the algorithm dynamically estimates the offset $\Delta^*$ based on the match between the decoding result and the received sequence, so as to \emph{minimize the influence of offset errors} on subsequent decoding and to reduce error propagation. This process is iteratively applied until the entire received sequence has been fully decoded.

\begin{algorithm}[!t]
\caption{Decode}
\label{alg:decode}
\renewcommand{\algorithmicrequire}{\textbf{Input:}}
\renewcommand{\algorithmicensure}{\textbf{Output:}}
\begin{algorithmic}[1]
\REQUIRE Model $\G$, history $h$, received token sequence $\hat{T}$, sample size $k$, distance constraint $d_\mathcal{T}$, block length $n_l$, windows size $w$, seed $\mathcal{S}$.
\ENSURE Reconstructed bits $\hat{s}$

\STATE $\hat{s} \leftarrow \emptyset$, $pos \leftarrow 0$
\STATE $\mathcal{R} \leftarrow \textsc{PRG.Init($\mathcal{S}$)}$
\WHILE{$len(\hat{T}) > pos$}
    \STATE $\B \leftarrow \textsc{CodebookConstruction}(\G,h,k,n_l,\mathcal{R})$ 
    \STATE $\mathcal{O} \leftarrow \textsc{DistanceConstraintGrouping}(\B,d_\mathcal{T})$ 
    \STATE $\msg \leftarrow \textsc{AdaptiveCoding}(\mathcal{O})$ 
    \STATE $\tilde{t} \leftarrow \hat{T}[pos:pos+nc+w]$ 
    \STATE $t \leftarrow\arg\min_{c \in \B} \dist\bigl(\tilde{t}, c\bigr)$
    \STATE $\Delta^\star \leftarrow \textsc{CalOffset($t,\tilde{t}$)}$
    \STATE $O \leftarrow \textsc{{Find($t,\mathcal{O}$)}}$
    \STATE $t^* \leftarrow \textsc{InvCDF}(O^\star,\mathcal{R}(\cdot))$
    \STATE $p \leftarrow \msg[O^\star]$
    \STATE $\hat{s} \leftarrow \hat{s} \,||\, p; \quad h \leftarrow h \,||\, t^*$
    \STATE $pos \leftarrow pos + n_l + \Delta^\star$
\ENDWHILE
\RETURN  $\hat{s}$
\end{algorithmic}
\end{algorithm}

\noindent \textbf{Batch Processing.}
The $k$ candidate sequences required for codebook construction are inherently amenable to parallel generation rather than strictly sequential execution. Taking advantage of the batch inference capability commonly supported by large language models~\cite{kwon2023efficientmemorymanagementlarge}, multiple distinct historical contexts $h$ can be packed into a single batch and \emph{processed in parallel} during model inference, thereby substantially improving the efficiency of codebook generation. 

In addition, since identical model states may recur during the generation process, the system incorporates a \emph{caching mechanism} that directly reuses previously computed probability distributions. This avoids redundant accesses to the generative model and further reduces the overall computational overhead.

\subsection{Security}

\begin{theorem}
\label{thm:computational-security}
Let $\Pi$ be a steganographic scheme in which the encoder $\Enc$ uses distance-constrained encoding and the decoder $\Dec$ employs minimum distance decoding, and if the encoding parameters are derived from a shared secret key $\mathsf{sk}$ via a cryptographically secure pseudorandom generator $\mathsf{PRG}$, then the system $\Pi$ is \textbf{computationally secure} provided that $\mathsf{PRG}$ is computationally indistinguishable from a uniform random source under security parameter $\lambda$.
\end{theorem}

The theorem above establishes that, when instantiated with a cryptographically secure PRG, the steganographic text generated by Alkaid is \emph{computationally indistinguishable} from the model's original random sampling output. The full proof is given in Appendix~\ref{appendix:computational-security}. The actual steganographic test results are in Appendix~\ref{appendix:testresult}.

\section{Evaluation}
\subsection{Experimental Setup}

\noindent \textbf{Generative Model Setup.} We focus on mainstream large language models (LLMs) commonly used today. Our experiments are primarily conducted on Qwen2.5-7B~\cite{qwen2025qwen25technicalreport}, LLaMA-3-8B~\cite{grattafiori2024llama3herdmodels}, GLM-4.5-9B~\cite{5team2025glm45}, and Mistral-7B~\cite{jiang2023mistral7b}, with Qwen2.5-7B serving as the default model for baseline evaluations. During generation, we employ nucleus (top-p) sampling~\cite{holtzman2019curious} at each step, with a default $p=0.95$. Specifically, as STEAD~\cite{qistead} is only applicable to diffusion large models, the model we equipped for it is Dream-7B~\cite{ye2025dream}.

\noindent \textbf{Hardware.} All of the experiments are conducted on a server with 25 vCPU Intel(R) Xeon(R) Platinum 8470Q, and a RTX 5090(32GB).

\noindent \textbf{Metrics and evaluation methods.} 
We evaluate the proposed method along four key dimensions: 

\noindent \textbf{\emph{(i)}} \textit{Robustness}: Measured by the decoding success rate. Only cases where the decoded message exactly matches the original are considered successful. We also report \(d_{\min}\), the minimum edit distance between the selected codeword and all other codewords in the alternative group of the codebook. A larger \(d_{\min}\) indicates greater error tolerance.

\noindent \textbf{\emph{(ii)}} \textit{Efficiency}: Measured by encoding and decoding throughput, reported in bits per second (bit/s) and tokens per second (token/s). Higher values indicate better efficiency.

\noindent \textbf{\emph{(iii)}} \textit{Capacity}: Measured by payload (bit/token) and entropy utilization. Higher values correspond to larger steganographic capacity.

\noindent \textbf{\emph{(iv)}} \textit{Usability}: Measured primarily by perplexity (PPL). Lower PPL signifies higher linguistic quality of the generated stegotext. Additionally, we measure max GPU memory consumption and demonstrate the system’s practical applicability by deploying it across multiple languages and models in diverse real-world scenarios. Furthermore, the use of different generative models is explored in Appendix~\ref{appendix:attempt}.

\begin{figure*}[!t]%
\centering
\includegraphics[width=0.95\linewidth]{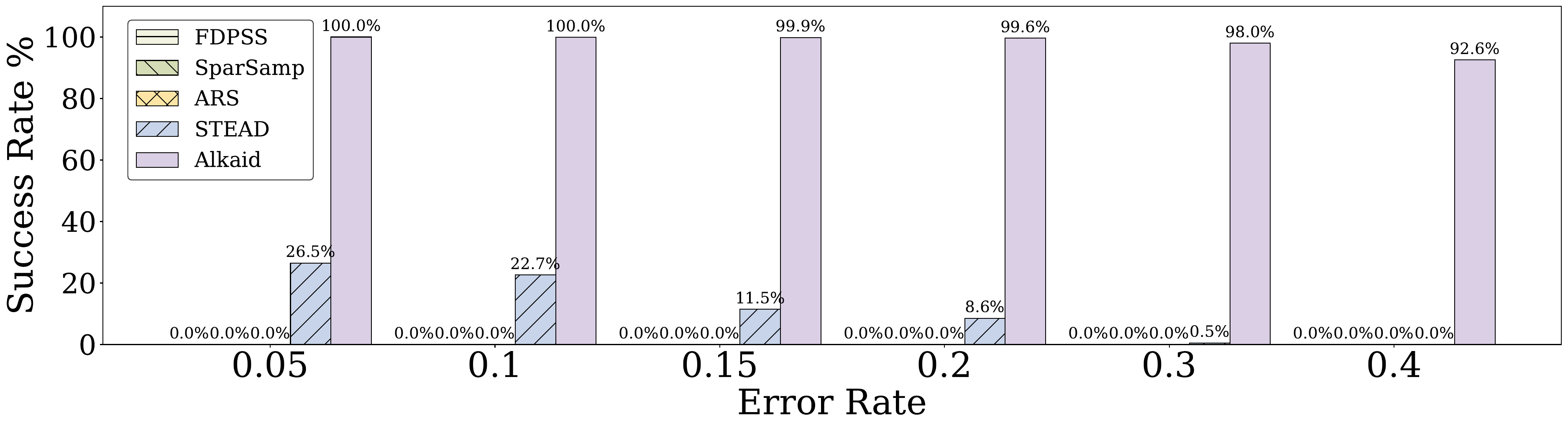}
    \caption{Decoding success rate comparison at different error rates on edit error channel $\Ce$.}
    \label{fig.exp1error}
\end{figure*}

\begin{figure*}[!t]%
\centering
\includegraphics[width=0.95\linewidth]{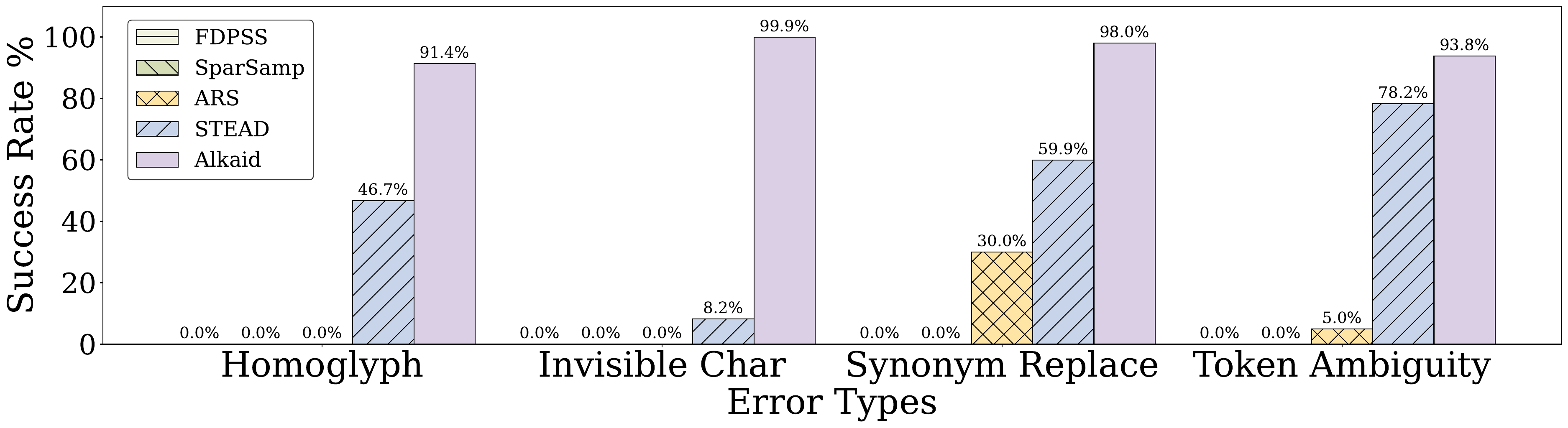}
    \caption{Decoding success rate comparison across different token-level errors.}
    \label{fig.exp1attack}
\end{figure*}

\subsection{Main results}

\noindent \textbf{Alkaid achieves resilience to edit errors}. By default, we assume that the system will encounter an error rate of 0.15 for the erroneous transmission channels, and we adopt a block size of $n_l = 20$. Therefore, the threshold $d_\mathcal{T} =0.15 * 20 = 6$. Additionally, we choose the sample size $k = 32$ to balance capacity and efficiency. For each error rate, 10,000 experiments were performed on an edit error channel.

As shown in the Fig.~\ref{fig.exp1error}, Alkaid consistently outperforms all other methods across random edit error rates from 0.05 to 0.4. It maintains a success rate above 99.6\% up to an error rate of 0.2, decreases only slightly to 98.0\% at 0.3, and remains as high as 92.6\% even at the maximum tested error rate of 0.4. In contrast, the state of the art (SOTA) methods FDPSS~\cite{liao2025framework}, SparSamp~\cite{wang2025sparsamp}, ARS~\cite{bai2025provably} degrade rapidly, with their success rates dropping to nearly zero. While STEAD~\cite{qistead} is robust to edit errors, it only achieves a success rate of 26.5\% at an error rate of 0.05, and its accuracy drops significantly as the error rate rises.  These results demonstrate that Alkaid is highly robust to random edit errors compared to existing methods.

\noindent \textbf{Alkaid achieves resilience to token-level errors}. Similar to previous studies~\cite{qu2025provably,kirchenbauer2023watermark}, we further conducted experiments on four typical token-level errors in the existing literature, with their core operations as follows: \textbf{\emph{(i)}} \emph{Homoglyph}, which replaces characters with visually similar glyphs; \textbf{\emph{(ii)}} \emph{Invisible Char}, which inserts imperceptible special characters into the text; \textbf{\emph{(iii)}} \emph{Synonym Replace}, which substitutes original words with their synonyms; \textbf{\emph{(iv)}} \emph{Token Ambiguity}, which misleads the tokenizer into decoding text into incorrect tokens. In all our experiments, each errors was set to perturb no more than 10\% of the tokens in the target text. 

As shown in Fig.~\ref{fig.exp1attack}, Alkaid demonstrates consistently high robustness across all four types of token-level errors, achieving success rates of 91.4\% under homoglyph errors, 99.9\% under invisible character errors, 98.0\% under synonym replacement errors, and 93.8\% under token ambiguity errors. By contrast, FDPSS and SparSamp perform poorly with a complete failure (0\% success rate). ARS also collapses under nearly all conditions, yielding success rates at or near 0\%; it only attains marginal performance under synonym replacement errors, due to its intrinsic resilience to lexical substitutions. STEAD exhibits moderate resilience, with success rates of 59.9\% under synonym replacement errors and 78.2\% under token ambiguity errors, yet it still lags far behind Alkaid across all scenarios. These results confirm that Alkaid is substantially more robust to diverse token-level perturbations than all state-of-the-art methods.

\begin{table*}[!t]
\renewcommand{\arraystretch}{1.2}
\caption{Comparison of efficiency and capacity of provably secure robust steganography.}
\label{tab:capacityoverheadcom}
\begin{tabularx}{\textwidth}{ Z|ZZZZZZZ}
\toprule
\makecell{\multirow{3}{*}{\textbf{Methods}}} &
\multicolumn{3}{c}{\textbf{Efficiency}} &
\multicolumn{2}{c}{\textbf{Capacity}} &
\multicolumn{2}{c}{\textbf{Usability}} \\
\cmidrule(lr){2-4} \cmidrule(lr){5-6} \cmidrule(l){7-8}
 & \makecell{Encoding$\uparrow$ \\(Token/s)} & \makecell{Encoding$\uparrow$ \\(Bit/s)} & \makecell{Decoding$\uparrow$ \\(Bit/s)} & \makecell{Payload$\uparrow$ \\(Bit/Token)} & \makecell{Entropy$\uparrow$ \\ Utilization\%} & Entropy & PPL$\downarrow$ \\
\midrule
STEAD   & \textbf{71.5330}            & 2.6385           & 0.4889           & 0.0369        & 0.72\%              & 5.0632      & 40.1805 \\
ARS     & 18.6490            & 0.2570           & \textbf{106.087}         & 0.0136        & 0.71\%              & 1.6500      & 1.1260  \\
Alkaid  & 33.0203            & \textbf{6.7236}           & 6.7011           & \textbf{0.2045}        & \textbf{18.47\%}              & 1.1493      & 2.4458  \\
\bottomrule
\end{tabularx}
\end{table*}

\begin{table*}[!t]
\centering
\renewcommand{\arraystretch}{1.2}
\caption{Comprehensive evaluation of Alkaid across multiple models and sampling configurations.}
\label{tab:performance_comparison}
\begin{tabularx}{\textwidth}{P{1.45cm} | Z | ZZ ZZZ P{1.2cm}P{1.2cm} ZP{1.1cm}Z}
\toprule
\small\makecell{\multirow{3}{*}{\textbf{Models}}} &
\small\makecell{\multirow{3}{*}{\textbf{Top-p}}} &
\multicolumn{2}{c}{\textbf{Robustness}} &
\multicolumn{3}{c}{\textbf{Efficiency}} &
\multicolumn{2}{c}{\textbf{Capacity}} &
\multicolumn{3}{c}{\textbf{Usability}} \\
\cmidrule(lr){3-4} \cmidrule(lr){5-7} \cmidrule(lr){8-9} \cmidrule(l){10-12}
 &  & 
\makecell{$d_{\min}$$\uparrow$} & \small\makecell{Succ.$\uparrow$ \\Rate\%} &
\small\makecell{Enc.$\uparrow$ \\(T./s)} & \small\makecell{Enc.$\uparrow$ \\ (B./s)} & \small\makecell{Dec.$\uparrow$ \\(B./s)} &
\small\makecell{Payload$\uparrow$ \\ (B./T.)} & \small\makecell{Entropy$\uparrow$\\ Util.\%} &
\small\makecell{GPU$\downarrow$ \\(GB)} & \small\makecell{Entropy} & \small\makecell{PPL$\downarrow$} \\
\midrule
\multirow{4}{*}{\textsc{Qwen}} 
& 0.8   & 10.547 & 99.9\% & 35.212 & 6.086 & 6.065 & 0.174 & 21.0\% & 16.468 & 0.881 & 1.973 \\
& 0.9   & 11.312 & 100\%  & 34.000 & 6.474 & 6.452 & 0.192 & 18.0\% & 16.468 & 1.145 & 2.495 \\
& 0.95  & 11.815 & 100\%  & 33.280 & 6.591 & 6.570 & 0.199 & 16.0\% & 16.468 & 1.357 & 3.101 \\
& 1.0   & 12.428 & 100\%  & 18.675 & 3.857 & 3.849 & 0.207 & 13.3\% & 16.468 & 1.703 & 4.898 \\
\midrule
\multirow{4}{*}{\textsc{Mistral}} 
& 0.8   & 9.734  & 99.9\% & 35.136 & 4.122 & 4.108 & 0.119 & 22.2\% & 25.101 & 0.571 & 1.632 \\
& 0.9   & 10.151 & 100\%  & 33.660 & 4.865 & 4.849 & 0.147 & 20.9\% & 25.101 & 0.769 & 2.048 \\
& 0.95  & 10.442 & 100\%  & 33.021 & 5.057 & 5.041 & 0.155 & 19.3\% & 25.101 & 0.888 & 2.538 \\
& 1.0   & 10.943 & 100\%  & 29.432 & 4.995 & 4.981 & 0.171 & 15.9\% & 25.101 & 1.226 & 5.539 \\
\midrule
\multirow{4}{*}{\textsc{Llama}}  
& 0.8   & 9.109  & 99.9\% & 33.182 & 3.464 & 3.452 & 0.105 & 21.9\% & 25.101 & 0.474 & 1.421 \\
& 0.9   & 9.631  & 99.9\% & 32.004 & 4.189 & 4.176 & 0.132 & 21.1\% & 25.101 & 0.626 & 1.594 \\
& 0.95  & 9.984  & 100\%  & 31.183 & 4.593 & 4.579 & 0.149 & 20.0\% & 25.101 & 0.758 & 1.812 \\
& 1.0   & 10.221 & 100\%  & 20.307 & 3.348 & 3.341 & 0.167 & 17.8\% & 25.101 & 0.969 & 2.456 \\
\midrule
\multirow{4}{*}{\textsc{Glm}}    
& 0.8   & 12.476 & 100\%  & 23.417 & 4.799 & 4.788 & 0.204 & 12.5\% & 25.101 & 1.742 & 4.666 \\
& 0.9   & 13.587 & 100\%  & 22.972 & 4.933 & 4.922 & 0.214 & 9.7\%  & 25.101 & 2.408 & 9.335 \\
& 0.95  & 14.368 & 100\%  & 22.332 & 4.932 & 4.921 & 0.221 & 8.0\%  & 25.101 & 3.009 & 20.556 \\
& 1.0   & 15.037 & 100\%  & 14.359 & 3.186 & 3.181 & 0.222 & 6.1\%  & 25.101 & 3.922 & 66.441 \\
\bottomrule
\end{tabularx}
\end{table*}

\noindent \textbf{Alkaid outperforms other robust steganography in terms of efficiency and steganographic capacity.} The results are shown in Table~\ref{tab:capacityoverheadcom}.
Although STEAD achieves a leading token encoding speed of 71.533 tokens per second due to its use of the Dream-7B diffusion model architecture, its actual steganographic encoding efficiency is only 2.6385 bits per second, with a payload capacity as low as 0.0369 bits per token and an entropy utilization rate of less than 1\%, resulting in limited overall steganographic capability. ARS is designed with an efficient decoding process, achieving a decoding speed of 106.087 bits per second, but its encoding efficiency is merely 0.257 bits per second, and its payload capacity is 0.0136 bits per token, similarly suffering from low capacity utilization. In contrast, Alkaid achieves an encoding efficiency of 6.7236 bits per second, far surpassing the other methods. Its payload capacity reaches 0.2045 bits per token, with an entropy utilization rate of 18.47\%, indicating its ability to conceal more information per unit time and significantly enhance the utilization efficiency of text entropy. In terms of usability, all methods maintain reasonable text perplexity, ensuring that model functionality remains unaffected. In summary, Alkaid outperforms the other methods in both encoding efficiency and steganographic capacity while maintaining good usability, leading to the conclusion that it achieves dual leadership in efficiency and capacity for robust steganography.

\noindent \textbf{Alkaid demonstrates high usability across diverse configurations.}
As shown in Table~\ref{tab:performance_comparison}, Alkaid exhibits consistent and adaptable performance across diverse sampling configurations, as evidenced by the comprehensive performance data of four typical large language models: Qwen, Mistral, Llama and GLM. By default, we set the configuration parameters as 
$(d_\mathcal{T},k,n_l) = (6,32,20)$. For the robustness test, we configured the edit error channel  with an error rate of 0.15. Across all models and Top-p settings from 0.8 to 1.0, the system maintains robust encoding speeds in terms of bits per second, with peaks such as 6.723 bits per second for Alkaid in earlier comparisons, and strong payload capacities ranging up to 0.222 bits per token in configurations like Glm with Top-p at 1.0. More importantly, success rates for message recovery reach 100\% in nearly all scenarios, supported by high minimum distance values generally above 10 in many settings, reflecting strong resilience against edit errors.

A clear pattern emerges across configurations: as Top-p increases, payload capacity tends to rise slightly, while entropy utilization generally decreases. For instance, in the Qwen model, payload increases from 0.174 to 0.207 bits per token as Top-p goes from 0.8 to 1.0, yet entropy utilization drops from 21.0\% to 13.3\%. Though the capacity experiences a reduction, this is indicative of enhanced robustness, which is reflected in the rise of $d_{\min}$ from 10.547 to 12.428. Encoding speed in tokens per second typically remains stable until Top-p approaches 1.0, where a noticeable decline occurs, as seen with Qwen dropping from over 33 tokens per second to 18.675. Text naturalness, measured by perplexity, remains within practical limits under most settings, staying below 5.0 for Top-p values up to 0.95 across models, though it can rise significantly at Top-p 1.0, particularly with certain models like Glm.
In summary, Alkaid exhibits high practical usability and reliable adaptability, balancing encoding speed, steganographic capacity, and text quality effectively across diverse models and sampling parameters. Its consistent performance in success rate and robustness, coupled with stable efficiency and capacity trends under varying conditions, confirms its suitability for a wide range of application scenarios.

\begin{figure*}[!t]

    \centering
    \subfloat[ Edit Distance Distribution. ($d_{min} > d_{\mathcal{T}}$, lower limit constraint)]{
        \label{fig.diffparameter_dmin}
        \includegraphics[width=0.47\linewidth]{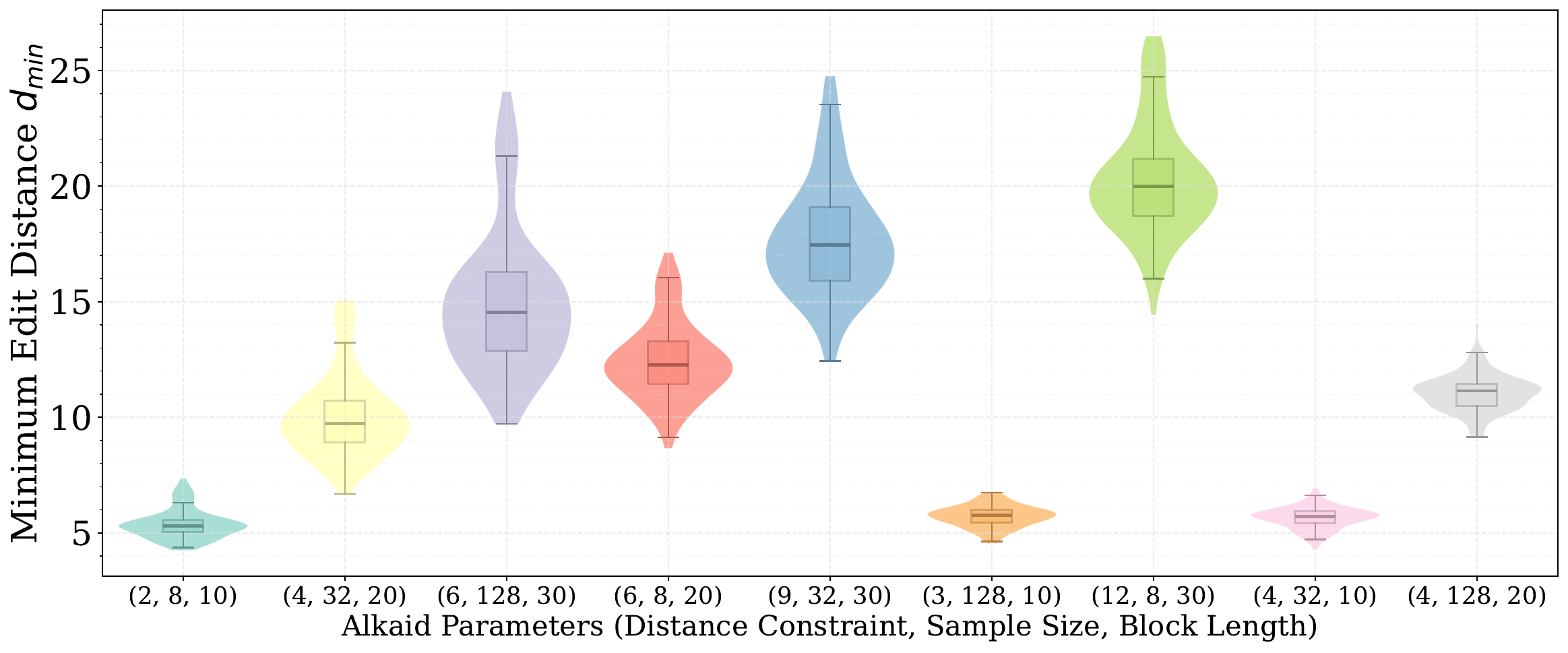}
        }
   \subfloat[ Capacity Distribution (bits). ($C\le \log_2k$, upper limit constraint)]{
        \label{fig.diffparameter_bits}
        \includegraphics[width=0.47\linewidth]{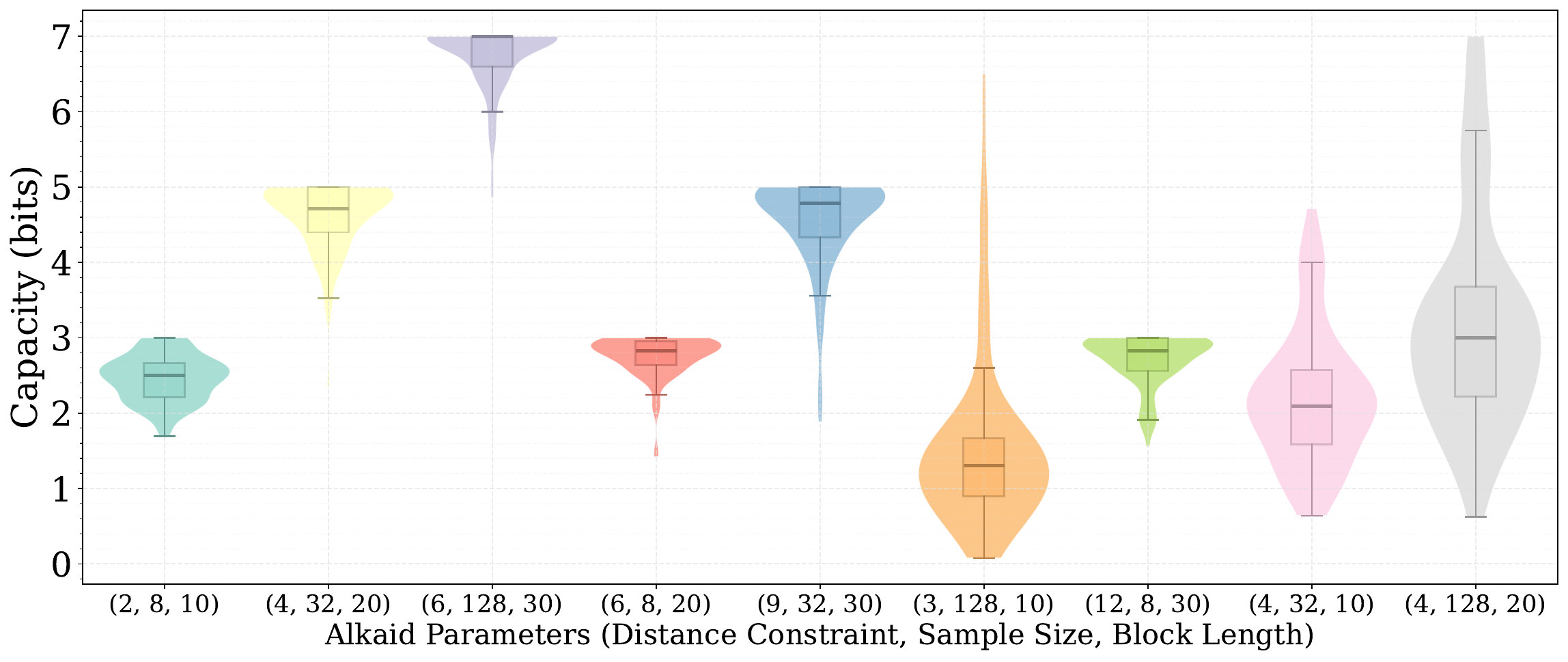}
        }
\caption{Trade-off between robustness and capacity under different configuration parameters ($d_\mathcal{T}$, $k$, $n_l$).}
\label{fig.diffparameter}
\end{figure*}

\begin{table*}[!t]
\centering
\renewcommand{\arraystretch}{1.2}
\caption{Performance of Alkaid under different configuration parameters ($d_\mathcal{T}$, $k$, $n_l$).}
\label{tab:config_performance}
\begin{tabularx}{\textwidth}{@{} c c c | ZZZ Z Z ZZ @{}}
\toprule
\multicolumn{3}{c|}{\textbf{Parameters}} &
\multicolumn{3}{c}{\textbf{Efficiency}} &
\multicolumn{2}{c}{\textbf{Capacity}} &
\multicolumn{2}{c}{\textbf{Usability}} \\
\cmidrule(r){1-3} \cmidrule(lr){4-6} \cmidrule(lr){7-8} \cmidrule(l){9-10}
$d_\mathcal{T}$ & $k$ & $n_l$ &
\makecell{Encoding$\uparrow$ \\ (Token/s)} &
\makecell{Encoding$\uparrow$ \\ (Bit/s)} &
\makecell{Decoding$\uparrow$ \\ (Bit/s)} &
\makecell{Payload$\uparrow$ \\ (Bit/Token)} &
\makecell{Entropy$\uparrow$ \\ Utilization\%} &
Entropy & PPL$\downarrow$ \\
\midrule
2   & 8   & 10  & 49.3844 & 11.9794 & 11.9696 & 0.2427 & 22.29\% & 1.1353 & 2.4600 \\
4   & 32  & 20  & 32.8736 & 7.3679  & 7.3433  & 0.2249 & 20.37\% & 1.1577 & 2.4686 \\
6   & 128 & 30  & 10.1460 & 2.1205  & 2.1085  & 0.2095 & 19.11\% & 1.1686 & 2.4337 \\
6   & 8   & 20  & 47.6695 & 6.4399  & 6.4320  & 0.1351 & 12.04\% & 1.1728 & 2.4893 \\
9   & 32  & 30  & 30.3108 & 4.4874  & 4.4686  & 0.1483 & 13.25\% & 1.1592 & 2.3825 \\
3   & 128 & 10  & 17.4250 & 2.1644  & 2.1550  & 0.1311 & 11.18\% & 1.0837 & 2.5654 \\
12  & 8   & 30  & 46.9241 & 4.1630  & 4.1562  & 0.0887 & 7.83\%  & 1.1882 & 2.5934 \\
4   & 32  & 10  & 36.7648 & 7.7365  & 7.7181  & 0.2125 & 18.30\% & 1.1294 & 2.4358 \\
4   & 128 & 20  & 11.5783 & 1.6014  & 1.5942  & 0.1411 & 12.68\% & 1.0793 & 2.3259 \\
\bottomrule
\end{tabularx}
\end{table*}

\subsection{Deep Dive}
\noindent \textbf{Different parameter variations of Alkaid exert distinct impacts on system performance.} To further explore the impact of each parameter of Alkaid on the experimental results, we conducted controlled experiments with different configuration parameters, and the relevant results are shown in Fig.~\ref{fig.diffparameter} and Table~\ref{tab:config_performance} respectively.

\begin{table*}[!t]
\centering
\renewcommand{\arraystretch}{1.2}
\caption{The impact of optimization strategies on the performance of Alkaid.}
\label{tab:optimization_strategies}
\begin{tabularx}{\textwidth}{ Z|ZZZZZZZ}
\toprule
\makecell{\multirow{3}{*}{\textbf{Methods}}} &
\multicolumn{4}{c}{\textbf{Efficiency}} &
\multicolumn{3}{c}{\textbf{Usability}} \\
\cmidrule(lr){2-5}  \cmidrule(l){6-8}
 & \makecell{Encoding$\uparrow$ \\(Token/s)} & \makecell{Encoding$\uparrow$ \\(Bit/s)} & \makecell{Decoding$\uparrow$ \\(Bit/s)} & \makecell{Density$\uparrow$ \\\small{(Token/s/GB)}} & \makecell{Max GPU$\downarrow$ \\ (GB)} & Entropy & PPL$\downarrow$ \\
\midrule
None           & 1.9892  & 0.3978 & 0.3977 & 0.1385 & 14.5295 & 1.1543 & 2.4762 \\
Caching        & 2.7197  & 0.5502 & 0.5500 & 0.1894 & 14.4682 & 1.1632 & 2.4185 \\
Batching       & 33.0656 & 6.7329 & 6.7104 & 2.0434 & 17.7603 & 1.1493 & 2.4458 \\
All            & 33.0280 & 6.7253 & 6.7026 & 2.0411 & 17.7603 & 1.1493 & 2.4458 \\
\bottomrule
\end{tabularx}
\end{table*}

\begin{table*}[!t]
\centering
\renewcommand{\arraystretch}{1.2}
\caption{Performance of Alkaid under different batch sizes.}
\label{tab:batch_performance}
\begin{tabularx}{\textwidth}{ Z|ZZZZZZZ}
\toprule
\makecell{\multirow{3}{*}{\textbf{Batch Sizes}}} &
\multicolumn{4}{c}{\textbf{Efficiency}} &
\multicolumn{3}{c}{\textbf{Usability}} \\
\cmidrule(lr){2-5}  \cmidrule(l){6-8}
 & \makecell{Encoding$\uparrow$ \\(Token/s)} & \makecell{Encoding$\uparrow$ \\(Bit/s)} & \makecell{Decoding$\uparrow$ \\(Bit/s)} & \makecell{Density$\uparrow$ \\\small{(Token/s/GB)}} & \makecell{Max GPU$\downarrow$ \\ (GB)} & Entropy & PPL$\downarrow$ \\
\midrule
16   & 17.0894 & 3.4934 & 3.4873 & 1.0815 & 16.7625  &  1.1712 & 2.4693 \\
32   & 24.3070 & 4.9088 & 4.8966 & 1.5098 & 18.0251  &  1.1745 & 2.5060 \\
64   & 32.5706 & 6.6482 & 6.6263 & 2.0131 & 17.6177 &  1.1507 & 2.4478 \\
128  & 33.0286 & 6.7249 & 6.7025 & 2.0411 & 17.7603 &  1.1493 & 2.4458 \\
144  & 33.0268 & 6.7254 & 6.7027 & 2.0410 & 17.7603  &  1.1493 & 2.4458 \\
\bottomrule
\end{tabularx}
\end{table*}

\emph{The distance constraint \(d_\mathcal{T}\) serves as the core variable governing the security and robustness of the system.} Data from Figure 1 clearly demonstrate that an increase in \(d_\mathcal{T}\) directly and significantly elevates the minimum edit distance, denoted as \(d_{\min}\), between the selected codeword and candidate codewords from other groups within the codebook. For instance, under configurations with roughly comparable values for the sample size \(k\) and block length \(n_l\), raising \(d_\mathcal{T}\) from 2, corresponding to configuration (2,8,10), to 12, corresponding to configuration (12,8,30), causes \(d_{\min}\) to surge sharply from 5.5 to 21.0. In another comparative case, with \(k\) fixed at 32 and \(n_l\) at 20, increasing \(d_\mathcal{T}\) from 4 to 6, that is, from configuration (4,32,20) to (6,8,20), leads to a jump in \(d_{\min}\) from 9.0 to 18.0. This reaffirms that this parameter establishes a quantifiable, monotonically increasing safety lower bound for the required textual divergence in the output. However, the comprehensive performance data in Table 3 reveal that this enhancement in robustness is accompanied by significant and consistent system overhead. In the first comparative case mentioned, the encoding efficiency measured in bits per second plummets from 11.9794 to 4.1630, while the effective payload measured in bits per token contracts substantially from 0.2427 to 0.0887. In the second comparative case, the encoding efficiency correspondingly decreases from 7.3679 to 6.4399 bits per second, and the payload drops from 0.2249 to 0.1351 bits per token. Consequently, \(d_\mathcal{T}\) represents the starting point for performance trade-offs: increasing it directly strengthens security but inevitably does so at the cost of reduced encoding speed and diminished steganographic capacity.

\emph{The sample size \(k\) acts as the core variable for regulating the system's capacity.} The capacity distributions illustrated in Figure 2 indicate that increasing \(k\) aims to expand the candidate sampling space, thereby theoretically exploring a higher upper limit for information embedding. Taking the configuration with fixed parameters \(d_\mathcal{T}=4\) and \(n_l=20\) as an example, when \(k\) is increased from 32 to 128, the maximum capacity rises from 4.9 bits to 6.5 bits, validating its design intent to enhance capacity potential. Similarly, within the context of \(d_\mathcal{T}=6\) and \(n_l=30\), comparing configuration (6,8,20) with (6,128,30) shows the third quartile of the capacity distribution increasing from 3.2 bits to 6.9 bits. However, the realization of this potential is constrained by security requirements and incurs a substantial computational burden. In the first comparative case with fixed other parameters, the median capacity does not rise with the increase in \(k\) but instead decreases from 4.9 bits to 3.1 bits, accompanied by a more dispersed distribution. This suggests that under stringent distance constraints, a larger search space may lead to increased outcome uncertainty rather than a stable improvement in effective payload. More critically, the computational cost is high. The encoding efficiency plummets from 7.3679 bits per second to 1.6014 bits per second in the first case. In the second comparative case, efficiency also sees a significant decline from 6.4399 bits per second to 2.1205 bits per second. Therefore, \(k\) is a parameter that requires careful trade-off consideration: it offers the possibility of raising the capacity ceiling, but its actual gain is strictly limited by security constraints and comes at an exponential increase in computational resource consumption.

\emph{The block length \(n_l\) is a parameter that simultaneously influences the robustness boundary, the capacity baseline, and computational efficiency.} Theoretically, \(n_l\) defines the length of the basic textual unit for each encoding operation, directly affecting the theoretical upper limit of information that can be carried per operation. The data show that, with other parameters held constant, increasing \(n_l\) can support a higher effective payload, although this trend may be reversed when combined with a high \(d_\mathcal{T}\). For example, comparing configuration (6,8,20) with (12,8,30), against the backdrop of a substantial increase in \(d_\mathcal{T}\), the payload further decreases from 0.1351 bits per token to 0.0887 bits per token despite \(n_l\) increasing from 20 to 30. However, under relatively moderate settings for \(d_\mathcal{T}\) and \(k\), such as \(d_\mathcal{T}=4\) and \(k=32\), increasing \(n_l\) from 10 to 20, moving from configuration (4,32,10) to (4,32,20), results in a slight rise in payload from 0.2125 to 0.2249 bits per token, while the minimum edit distance \(d_{\min}\) increases significantly from 4.0 to 9.0. Simultaneously, an increase in \(n_l\) linearly or non-linearly expands the search dimensionality of the encoding process, leading to a notable rise in computational complexity. This is corroborated by the trend observed in the encoding speed measured in tokens per second, which decreases from 36.7648 to 32.8736 in one comparison and from 47.6695 to 46.9241 in another, as \(n_l\) increases. Furthermore, \(n_l\) determines the textual scope over which the distance constraint \(d_\mathcal{T}\) is applied. A larger \(n_l\) means that the divergence requirement is enforced over longer contiguous text segments, thereby granting higher robustness against concentrated errors.


\noindent\textbf{The optimization strategy used by Alkaid has significantly improved its efficiency.} The performance impact of different optimization strategies is detailed in Table~\ref{tab:optimization_strategies}. Specifically, without any optimization applied, the system's encoding efficiency is only 1.9892 tokens per second or 0.3978 bits per second. Introducing the caching strategy results in a preliminary performance gain, increasing the encoding speed to 2.7197 tokens per second, an improvement of approximately 37 percent. However, the true leap in efficiency stems from the application of the batching strategy. With batching enabled, the encoding speed soars to 33.0656 tokens per second, representing an increase of over sixteen times compared to the baseline version. The effective encoding speed measured in bits per second also reaches 6.7329, showing a similarly dramatic increase. From the perspective of resource utilization, the efficiency gains are not achieved at the cost of a linear increase in resource consumption. The computational density metric, measured in tokens per second per gigabyte of GPU memory, increases dramatically from 0.1385 without optimization to 2.0434 with batching. This signifies that the encoding capability per unit of GPU memory is enhanced by nearly fifteen times, highlighting the exceptional effectiveness of the optimization strategies in improving hardware utilization efficiency. Concurrently, the system's usability metrics remain stable. Although batching requires more memory for parallel operations, increasing the peak GPU memory usage from approximately 14.5 GB to 17.8 GB, the quality of the generated text is unaffected. The entropy value stays consistent, and the perplexity remains at a consistently low level around 2.45, indicating that the naturalness and fluency of the text are preserved.
In conclusion, the optimization strategies for Alkaid, particularly the batching technique, have successfully increased its encoding efficiency by more than an order of magnitude while also achieving higher hardware utilization density. This dramatic improvement in efficiency is accomplished without compromising the usability of the generated text, underscoring the engineering efficacy and practical value of the optimization approach.

\textbf{The performance of Alkaid under different batch sizes.} Table~\ref{tab:batch_performance} details the performance of Alkaid under varying batch sizes. The data analysis indicates that increasing the batch size significantly enhances the system's operational efficiency, with performance gains stabilizing after reaching a certain scale, while all configurations maintain highly consistent text generation quality. In terms of operational efficiency, as the batch size increases from 16 to 128, encoding speed measured in tokens per second rises steadily from 17.0894 to 33.0286, a 93 percent improvement, and the effective bit encoding speed doubles from 3.4934 to 6.7249 bits per second. Notably, the performance improvement exhibits diminishing marginal returns. The most substantial gain occurs when scaling from 16 to 64, after which the performance curve plateaus, with nearly identical results for batch sizes 128 and 144. This suggests that a batch size between 64 and 128 likely approaches a performance saturation point under the current hardware and model configuration. The computational density metric shows a corresponding growth from 1.0815 to 2.0411 tokens per second per gigabyte of GPU memory, indicating that larger batches more fully utilize parallel computing units and substantially improve hardware efficiency. Meanwhile, the peak GPU memory consumption remains relatively stable. Regarding usability, key metrics remain excellent and stable across all configurations: text entropy stays within the narrow range of 1.15 to 1.17, and perplexity remains low, between 2.45 and 2.51. This demonstrates that batch processing optimization achieves a major throughput increase without compromising the semantic coherence or linguistic quality of the generated text. In summary, adjusting the batch size allows Alkaid to effectively improve encoding efficiency and hardware utilization density without affecting text usability.

\section{Conclusion}

Conventional Provably Secure Steganography (PSS) offers strong security guarantees but fails in practice due to its fragility to common edit errors. To bridge this gap between theory and robustness, we propose Alkaid, a system that embeds error resilience directly into the secure encoding process via distance-constrained encoding. By ensuring a minimum separation between the carrier sequences for different messages, Alkaid guarantees that a message remains uniquely decodable as long as the accumulated edit errors do not exceed a defined bound. We formally prove that this construction preserves the computational indistinguishability central to PSS while providing deterministic robustness against bounded distortions. Experiments demonstrate that Alkaid effectively withstands realistic edits, achieving a rigorous and practical unification of provable security and error resilience.

\bibliographystyle{ieeetr}
\bibliography{ref}

\appendices

\section{Proof}
\subsection{Proof of Theorem~\ref{thm:distance-constraint-security}} \label{appendix:proof_theoreticsecurity}

\begin{proof}
Let $M$ be the random message with conditional distribution $\Pr[M = m \mid \xi] = \frac{|\mathcal{O}^\xi_m|}{k}$. For any $\x \in \Xn$, by the law of total expectation, we have
\begin{equation}
\ppi(\x) = \Pr[\Enc(M,\xi) = \x] = \mathbb{E}_{\xi}\Bigl[ \Pr[\Enc(M,\xi) = \x \mid \xi] \Bigr].
\end{equation}

Given $\xi$, the  candidate sequences in the codebook $\B^\xi = \{\x^{(1)}, \dots, \x^{(k)}\}$ are independent and identically distributed according to $\ptheta$. Define indicator variables $\mathbf{1}_i(\x) = \mathbf{1}\{\x^{(i)} = \x\}$, then
\begin{equation}
t_\x(\xi) = \sum_{i=1}^k \mathbf{1}_i(\x),
\end{equation}
\begin{equation}
\quad \mathbb{E}[t_\x(\xi)] = \sum_{i=1}^k \Pr[\x^{(i)} = \x] = k \cdot \ptheta(\x).    
\end{equation}

The grouping $\{\mathcal{O}^\xi_m\}$ is determined by $\B^\xi$. When $\eta$ is sampled randomly, the encoder satisfies: given $M = m$, $\Enc(M,\xi)$ selects a sequence instance from $\mathcal{O}^\xi_m$ according to the prescribed probability distribution. Therefore,
\begin{equation}
\Pr[\Enc(M,\xi) = \x \mid \xi] 
\begin{aligned}[t]
    & = \sum_{m \in \msg^\xi} 
    \Pr[M = m \mid \xi] \\
    &\cdot \Pr[\Enc(M,\xi) = \x \mid M = m,\, \xi].
  \end{aligned}
\end{equation}

It follows that
\begin{equation}
\begin{aligned}[t]
\Pr[\Enc(M,\xi) = \x \mid \xi] &= \sum_{m: \x \in \mathcal{O}^\xi_m} \frac{|\mathcal{O}^\xi_m|}{k} \cdot \frac{t_\x(\xi)}{|\mathcal{O}^\xi_m|} \\
&= \sum_{m: \x \in \mathcal{O}^\xi_m} \frac{t_\x(\xi)}{k}.
\end{aligned}
\end{equation}

For fixed $\xi$, according to the grouping strategy, $\x$ belongs to a unique group $\mathcal{O}^\xi_{m^*}$, so the sum reduces to a single term:
\begin{equation}
\Pr[\Enc(M,\xi) = \x \mid \xi] = \frac{t_\x(\xi)}{k}.   
\end{equation}

Taking expectation over $\xi$ gives
\begin{equation}
\begin{aligned}[t]
\ppi(\x) &= \mathbb{E}_{\xi}\!\left[ \frac{t_\x(\xi)}{k} \right] = \frac{1}{k} \mathbb{E}[t_\x(\xi)] \\
&= \frac{1}{k} \cdot k \cdot \ptheta(\x) = \ptheta(\x).   
\end{aligned}
\end{equation}

Since $\x$ is arbitrary, we have $\ppi = \ptheta$. According to the Definition~\ref{def:information-theoretic-security}, the system is \textbf{perfectly secure}.
\end{proof}

\subsection{Proof of Theorem~\ref{thm:distance-constraint-robustness}}\label{appendix:dce-roubust}
\begin{proof}
Let $\x = \Enc(m,\xi)$ and $\tilde{\x} = \Ce(\x)$, and define the random variable $D = \dist(\x, \tilde{\x})$ representing the total number of edit operations introduced by the channel. According to the channel definition, let $V_i \in \{0,1\}$ indicate whether the $i$-th position undergoes a basic edit operation. Then $\{V_i\}_{i=1}^n$ are independent with $\mathbb{E}[V_i] = e$, and we have
\begin{equation}
D \le \sum_{i=1}^n V_i, \quad \mathbb{E}[D] \le e n.
\end{equation}

If $D < \frac{d_\mathcal{T}}{2}$, then for any $m' \neq m$, by the triangle inequality for edit distance,
\begin{equation}
\begin{aligned}
\dist(\tilde{\x}, \Enc(m', \xi)) 
&\ge \dist(\x, \Enc(m', \xi)) - \dist(\x, \tilde{\x}) \\
&> d_\mathcal{T} - \frac{d_\mathcal{T}}{2} 
= \frac{d_\mathcal{T}}{2}.
\end{aligned}
\end{equation}

At the same time, we have 
$
\dist(\tilde{\x}, \Enc(m, \xi)) \le \dist(\tilde{\x}, \x) < \frac{d_\mathcal{T}}{2},
$
so minimum distance decoding will necessarily output $m$. Therefore, a decoding error can occur only if $D \ge \frac{d_\mathcal{T}}{2}$, which gives
\begin{equation}
\Pr[\Dec(\tilde{\x}, \xi) \neq m] \le \Pr\!\left[D \ge \frac{d_\mathcal{T}}{2}\right].    
\end{equation}

Let $S = \sum_{i=1}^n V_i$, which is the sum of $n$ independent random variables taking values in $[0,1]$, with $\mathbb{E}[S] = e n$. Define
$
t = \frac{d_\mathcal{T}}{2} - e n > 0.
$
By Hoeffding's inequality, we have 
\begin{equation}
\Pr[S - \mathbb{E}[S] \ge t] \le \exp\!\left(- \frac{2 t^2}{n} \right).
\end{equation}

Since $D \le S$, it follows that
\begin{equation}
\begin{aligned}
\Pr\!\left[D \ge \frac{d_\mathcal{T}}{2}\right]
&\le \Pr\!\left[S \ge \tfrac{d_\mathcal{T}}{2}\right] \\
&= \Pr\!\left[S - e n \ge t\right] \\
&\le \exp\!\left(-\frac{2}{n} \left(\frac{d_\mathcal{T}}{2} - e n\right)^2 \right).
\end{aligned}
\end{equation}

Rewriting gives
\begin{equation}
\Pr[\Dec(\Ce(\Enc(m,\xi)), \xi) \neq m]
\le \exp\!\left(- \frac{(d_\mathcal{T} - 2 e n)^2}{2 n} \right),
\end{equation}
which completes the proof.
\end{proof}

\subsection{Proof of Theorem~\ref{thm:computational-security}}\label{appendix:computational-security}

\begin{proof}
During encoding, the system uses the pseudorandom generator $\mathsf{PRG}$ to derive the encoding parameters
$
\xi = (\zeta_1,\dots,\zeta_k,\eta),
$
from the shared seed $\mathsf{sk}$, where each component is generated deterministically from the output of $\mathsf{PRG}(\mathsf{sk})$. In practical implementation, we can segment the $\mathsf{PRG}$  randomly generated sequence into blocks and process them into encoding parameters.

We define the following hybrid experiments.

\emph{Experiment} $\mathsf{Exp}_0(\lambda)$: Sample a message $m \sim \pmsg$ and a shared seed $\mathsf{sk} \sim \{0,1\}^\lambda$. The encoder derives all encoding parameters $\xi$ from $\mathsf{PRG}(\mathsf{sk})$ and runs $\Enc(m,\xi)$ to generate the output sequence $\x$, which is then provided to the distinguisher $\mathcal{A}$.

\emph{Experiment} $\mathsf{Exp}_1(\lambda)$: Same as $\mathsf{Exp}_0$, but all random components of $\xi$ are replaced by independent uniform samples from their respective domains.

\medskip
By the computational indistinguishability of $\mathsf{PRG}$, there exists a negligible function $\negl(\cdot)$ such that for any PPT distinguisher $\mathcal{A}$,
\begin{equation}
\left| \Pr[\mathcal{A}(\mathsf{Exp}_0(\lambda))=1] - \Pr[\mathcal{A}(\mathsf{Exp}_1(\lambda))=1] \right| \le \negl(\lambda).  
\end{equation}

Next, consider the output distribution induced in $\mathsf{Exp}_1(\lambda)$. In this experiment, the components of the encoding parameters
$
\xi = (\zeta_1,\dots,\zeta_k,\eta)
$
are mutually independent and independently sampled according to their respective domains.

Specifically, during codebook construction, each component $\zeta_i$ can be represented as a random vector
$
\zeta_i = (r_{i,1},\dots,r_{i,n_l}),
$
where each element $r_{i,t}$ is independent and uniformly distributed. Hence, the joint distribution of $\xi$ satisfies all conditions regarding independence and distribution matching required in Theorem~\ref{thm:distance-constraint-security}.

\section{The Analysis of Robustness} \label{appendix:analysisofrobust}

\begin{figure*}[!t]%
\centering
\includegraphics[width=0.95\linewidth]{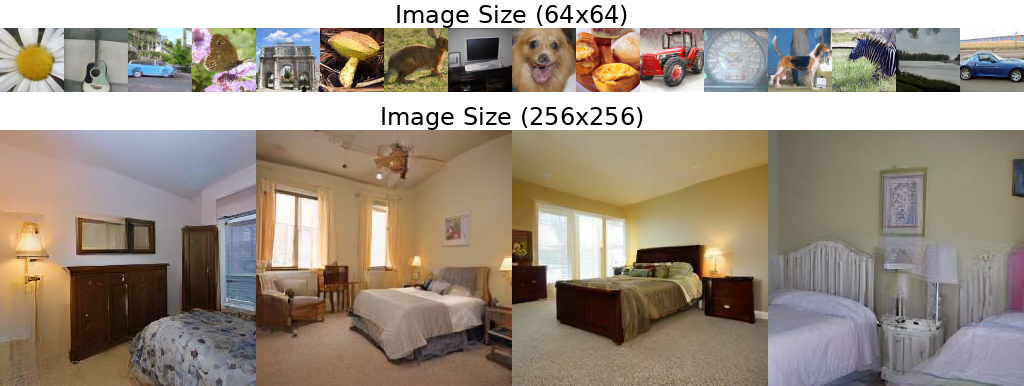}
    \caption{Samples of employing distance-constrained encoding on Consistency Models.}
    \label{fig.cd_generate}
\end{figure*}

\begin{figure}[!t]%
\centering
\includegraphics[width=0.85\linewidth]{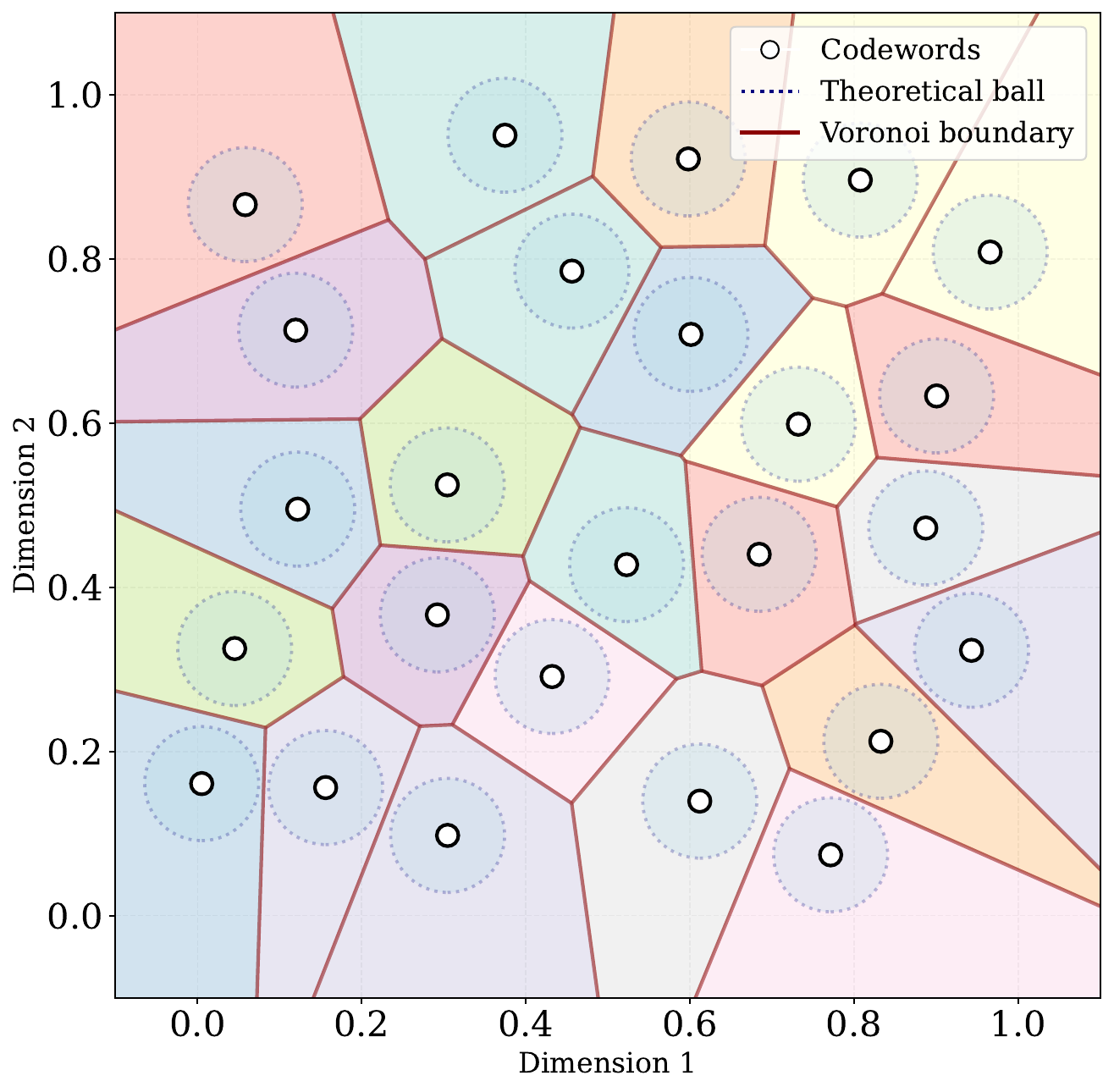}
    \caption{Actual Voronoi Decoding Regions.}
    \label{fig.voronoi}
\end{figure}

Moreover, the adaptive message encoding mechanism does not change the induced distribution over the message space $\msg^\xi$, so that
$
\pmsg^\xi(m) = \frac{|\mathcal{O}^\xi_m|}{k}.
$

Therefore, by Theorem~\ref{thm:distance-constraint-security}, the system is perfectly secure in $\mathsf{Exp}_1(\lambda)$. Thus, we have 
$
\ppi = \ptheta.
$
Consequently,
\begin{equation}
\Pr[\mathcal{A}(\mathsf{Exp}_1(\lambda))=1] = \Pr_{\x \sim \ptheta}[\mathcal{A}(1^\lambda,\x)=1].    
\end{equation}

Combining the above inequality with the triangle inequality gives
\begin{equation}
\begin{aligned}
\mathsf{Adv}_{\mathcal{A},\Pi}(\lambda)
&= \left| \Pr[\mathcal{A}(\mathsf{Exp}_0(\lambda))=1] - \Pr_{\x \sim \ptheta}[\mathcal{A}(1^\lambda,\x)=1] \right| \\
&\leq \negl(\lambda).
\end{aligned}
\end{equation}
which shows that the system $\Pi$ satisfies Definition~\ref{def:computational-indistinguishability-security} and is \textbf{computationally secure}.
\end{proof}

The robustness proof in Theorem~\ref{thm:distance-constraint-robustness} is based on a worst-case assumption where the minimum distance between codewords exactly equals the design threshold $d_\mathcal{T}$, yielding the corresponding theoretical bound. However, in Alkaid's practical implementation, the sample size $k$ is far smaller than the codebook space, and this sparsity renders the $d_\mathcal{T}$-derived theoretical bound conservative. To verify Alkaid's robustness under realistic configurations, we conducted a targeted experiment: 100 codewords were randomly sampled from a 100-ary alphabet of length 100, with a strict constraint enforcing the Hamming distance between any pair of codewords to exceed $d_\mathcal{T}$, while Euclidean distance was incorporated as a hard decision criterion in minimum-distance decoding. As shown in Figure~\ref{fig.voronoi}, all Voronoi cells expand outward and completely envelop the theoretical Hamming balls. Despite exhibiting irregular geometric shapes due to local density fluctuations from random sampling, every cell strictly exceeds the theoretical boundary. This phenomenon directly demonstrates that sparsity enlarges the actual decoding regions, causing the effective error-correction radius to universally surpass $d_\mathcal{T}/2$. Combined with the observation in Table~\ref{tab:performance_comparison} that $d_{\min}$ consistently exceeds $d_\mathcal{T}$, Alkaid's deployed decoding failure rate will be lower than the worst-case theoretical upper bound derived from $d_\mathcal{T}$, thereby providing substantive assurance for its reliable deployment in real-world noisy environments.

Furthermore, due to Alkaid's block-wise processing architecture, error-correction capability remains confined within individual blocks without cross-block support. Consequently, the system exhibits limited resilience against clustered errors concentrated in localized regions. Therefore, practical deployment requires careful selection of block size $n_l$ and sample size $k$ to balance computational efficiency and robustness.

\section{Steganographic test results} \label{appendix:testresult}

\begin{table}[!t]
\centering
\caption{Steganalysis results for Alkaid.}
\renewcommand{\arraystretch}{1.2}
\begin{tabular}{c c}
\toprule
Methods & Detection Rate\% \\
\midrule
FCN~\cite{yang2019fast}         & 50.11$\pm$0.22\% \\
R-BiLSTM-C~\cite{yang2020linguistic} & 49.82$\pm$0.43\% \\
LSTMATT~\cite{zou2020high}      & 50.05$\pm$0.49\% \\
\bottomrule
\end{tabular}

\label{tab:steganalysis}
\end{table}

To comprehensively evaluate the security of Alkaid, we empirically tested it using multiple mature steganalysis methods. The goal was to distinguish between randomly sampled text and stego text generated by Alkaid. We employed three deep learning-based approaches proven effective: FCN~\cite{yang2019fast} , R-BiLSTM-C~\cite{yang2020linguistic} , and LSTMATT~\cite{zou2020high} .

As shown in Table~\ref{tab:steganalysis}, the detection rate for Alkaid is close to 50\%. In steganalysis, deviations below 0.5\% are generally considered negligible. Our experimental results, which show an detection rate below 0.5\%, therefore indicate that current steganalysis methods perform no better than random guessing in identifying Alkaid-generated content. These findings provide empirical support for the security of Alkaid.

\section{Attempt on other generative models}\label{appendix:attempt}

\begin{table}[!t]
\centering
\renewcommand{\arraystretch}{1.2}
\caption{Performance of distance constrained encoding in different generative models.}
\begin{tabularx}{0.47\textwidth}{P{1.7cm}P{1.9cm}ZZ}
\toprule
Methods & Payload & Speed & Success Rate \\
\midrule
Dream~\cite{ye2025dream} & 0.13 bits/token & 0.11 bits/s & 100\%\\
CM-64x~\cite{song2023consistency} & 2.2e-3 bits/pixel & 5.91 bits/s & 100\%\\
CM-256x~\cite{song2023consistency} & 1.2e-3 bits/pixel & 0.63 bits/s & 100\%\\
\bottomrule
\end{tabularx}

\label{tab:attampttoother}
\end{table}

While prior research has achieved certain outcomes, it typically relies on generative models providing an explicit probability distribution. This requirement limits their applicability to  generative methods such as diffusion models and VAEs, as these models generally do not directly output probability distributions~\cite{jois2024pulsar,lee2025removal,peng2023stegaddpm}. In contrast, distance-constrained encoding only requires the model to be capable of generating samples, making it suitable for a wider range of generative models. To verify this, we conducted experimental evaluations on several generative models.

First, we employed Dream-7B~\cite{ye2025dream}, a diffusion language model based on a denoising mechanism. By controlling its random sampling process, we generated 256 sequences of length 64 with 300 sampling steps. The results, shown in Table~\ref{tab:attampttoother}, indicate that compared to the STEAD method in Table~\ref{tab:performance_comparison}, distance-constrained encoding shows no significant decrease in embedding capacity and rate. Due to the higher uncertainty in the generation process of diffusion models, the differences between samples are substantial, leading to a corresponding increase in decoding success rate.

Furthermore, we conducted experiments on Consistency Models~\cite{song2023consistency}. Since this model does not directly provide a probability distribution, we controlled the generation by adjusting the initial random noise. This still maintains security. We tested two model scales: one generating 256 images of size 64×64 and the other generating 256 images of size 256×256. As shown in Fig.~\ref{fig.cd_generate} and Table~\ref{tab:attampttoother}, the embedding rate and speed are higher in the smaller model but decrease in the larger model. However, due to significant differences between generated images, the decoding success rate remains high. This suggests that images inherently possess high entropy, providing a potential space for encoding. Nevertheless, direct global encoding is inefficient, and a block-based mechanism similar to Alkaid is needed to reduce the exponential complexity of distance-constrained encoding to linear, thereby improving practical efficiency.

These experiments demonstrate that distance-constrained encoding performs well across generative models of different types and scales. Its core advantage lies in relying only on sample generation capability, rather than specific probability distribution forms. This opens up possibilities for a unified distance-constrained encoding framework across various models. Future work could explore optimization strategies tailored to different model characteristics, such as blocking, pruning, and parallelization. Due to space limitations, this paper primarily focuses on efficient implementation in autoregressive large language models. \emph{Systematic optimization and theoretical refinement under different generative paradigms remain important directions for further research.}

\end{document}